\documentclass{article}

\usepackage{times}  
\usepackage{helvet}  
\usepackage{courier}  
\usepackage{url}  
\usepackage{graphicx}  
\usepackage{mathtools}
\usepackage{amsmath,amssymb,amsfonts}
\DeclareMathOperator*{\argmax}{arg\,max}
\DeclareMathOperator*{\argmin}{arg\,min}
\usepackage{caption}
\usepackage{amsthm,dsfont}
\newcommand{\ignore}[1]{}
\usepackage{color}
\usepackage{authblk}

\usepackage[vlined,linesnumbered,ruled,titlenotnumbered]{algorithm2e}
\usepackage{algorithm2e}

\newtheorem{theorem}             {Theorem}

\newtheorem{corollary}  [theorem]{Corollary}

\newtheorem{definition} [theorem]{Definition}

\usepackage{pdflscape}
\usepackage{adjustbox}
\usepackage{rotating}
\usepackage{capt-of}

\usepackage{booktabs} 
\usepackage{graphicx}
\usepackage[vlined,linesnumbered,ruled,titlenotnumbered]{algorithm2e}
\usepackage{mathtools}
\usepackage{amsthm,dsfont}

\usepackage{color}
\usepackage{hyperref}
\usepackage{balance}
\usepackage{xspace}
\usepackage{comment}
\usepackage{algpseudocode,booktabs,amsmath}
\usepackage{multirow, rotating}
\usepackage{mwe}
\usepackage{color}
\usepackage{blindtext}
\usepackage{amsmath,amssymb,amsfonts}
\usepackage{mathtools}

\DeclareMathOperator*{\bin}{bin}
\newcommand{\POMCwp}{POMC$^\text{wp}$}

\usepackage{subcaption}
\usepackage{rotating}

\usepackage{hyperref}
\usepackage{balance}
\usepackage{tikz}
\usetikzlibrary{shapes,decorations,patterns,positioning}
\usetikzlibrary{plotmarks}
\usepackage{url}
\usepackage{booktabs}
\usepackage{multirow}
\usepackage{algorithm2e}
\usepackage{changepage}

\DeclareSymbolFont{symbolsC}{U}{pxsyc}{m}{n}
\DeclareMathSymbol{\colonequals}{\mathrel}{symbolsC}{"42}

\usepackage{url}
\urldef{\mailsa}\path|{vahid.roostapour, aneta.neumann, frank.neumann}@adelaide.edu.au|

\usepackage[vlined,linesnumbered,ruled,titlenotnumbered]{algorithm2e}
\usepackage{algorithm2e}

\begin{document}
	
	\title{Pareto Optimization for Subset Selection with Dynamic Cost Constraints\footnote{A preliminary version of this article has been presented at the Thirty-Third {AAAI} Conference on Artificial Intelligence ({AAAI} 2019)~\cite{DBLP:conf/aaai/RoostapourN0019}}}

\author[1]{Vahid Roostapour\thanks{vahid.roostapour@adelaide.edu.au, vahid.roostapour@gmail.com}}
\author[1]{Aneta Neumann\thanks{aneta.neumann@adelaide.edu.au}}
\author[1]{Frank Neumann\thanks{frank.neumann@adelaide.edu.au}}
\author[2]{\\Tobias Friedrich\thanks{tobias.friedrich@hpi.de}}
\affil[1]{Optimisation and Logistics, School of Computer Science, The University of Adelaide, Adelaide, Australia}
\affil[2]{Chair of Algorithm Engineering, Hasso Plattner Institute, Potsdam, Germany}

\renewcommand\Authands{ and }
\maketitle

\begin{abstract}
	We consider the subset selection problem for function $f$ with constraint bound $B$ that changes over time. Within the area of submodular optimization, various greedy approaches are commonly used.  For dynamic environments we observe that the adaptive variants of these greedy approaches are not able to maintain their approximation quality. Investigating the recently introduced POMC Pareto optimization approach, we show that this algorithm efficiently computes a $\phi= (\alpha_f/2)(1-\frac{1}{e^{\alpha_f}})$-approximation, where $\alpha_f$ is the submodularity ratio of $f$, for each possible constraint bound $b \leq B$. Furthermore, we show that POMC is able to adapt its set of solutions quickly in the case that $B$ increases. Our experimental investigations for the influence maximization in social networks show the advantage of POMC over generalized greedy algorithms. We also consider EAMC, a new evolutionary algorithm with polynomial expected time guarantee to maintain $\phi$ approximation ratio, and NSGA-II with two different population sizes as advanced multi-objective optimization algorithm, to demonstrate their challenges in optimizing the maximum coverage problem. Our empirical analysis shows that, within the same number of evaluations, POMC is able to perform as good as NSGA-II under linear constraint, while EAMC performs significantly worse than all considered algorithms in most cases.
\end{abstract}


\section{Introduction}

Many combinatorial optimization problems consist of optimizing a given function under a given constraint. Constraints often limit the resources available to solve a given problem and may change over time. For example, in the area of supply chain management, the problem may be constrained by the number of vehicles available that may change during the optimization process due to vehicle failures and vehicles being repaired or added as a new resource. As another example, consider the problem of using a train to transport extracted materials form a number of mines to a port. In this example, the number of available wagons could change during the path because of different reasons. The question would be how to re-allocate the available wagons to transform the high priority materials first.

Evolutionary algorithms as general-purpose problem solvers have been widely used to tackle dynamic optimization problems theoretically and empirically~\cite{DBLP:conf/gecco/Yang15}. Many theoretical investigations considered the performance of evolutionary algorithms in the dynamic environments, where either the constraint or some properties of the search space might change during the optimization process. Neumann and Witt studied the behavior of the simplest evolutionary algorithm called (1+1)~EA on the dynamic makespan scheduling problem, in which the adversary can change the processing time of a job~\cite{DBLP:conf/ijcai/NeumannW15}.
They calculated the expected time of finding a good quality solution in different situations, where the algorithms either start from scratch and the dynamic changes happen frequently or start from a good quality solution and aim to find a new good quality solution after one change. The same analyses have also been considered on dynamic graph problems where a dynamic change might add or remove an edge to or from the graph~\cite{DBLP:conf/gecco/Bossek0PS19,POURHASSAN2019}. More recently, Shi et al. theoretically investigated the re-optimization time of multi-objective evolutionary algorithms in the dynamic environment for general linear functions~\cite{DBLP:journals/algorithmica/ShiSFKN19}. They proved the beneficial effect of populations in multi-objective approaches dealing with dynamic alterations. Their theoretical results have been extended to experimental studies in~\cite{roostapour2020pareto} by benchmarking dynamic knapsack problem. Roostapour et al. showed that the simple multi-objective algorithms that limit their population size by storing solutions based on their size perform significantly worse than the others.  

Submodular functions form an important class of problems as many important optimization problems can be modeled by them. A function is submodular when adding an element to a solution set be less beneficial than adding the same element to any of its subsets.
The area of submodular function optimization under given static constraints has been studied quite extensively in the literature. Nemhauser et al. considered maximizing nondecreasing submodular functions under cardinality constraint~\cite{nemhauser1981maximizing}. They proved that the greedy algorithm that iteratively adds a new element with a maximum marginal gain could maintain $(1-1/e)$-approximation. Following their results, Khuller et al. showed the unbounded approximation of the greedy approach in maximizing submodular functions under linear constraint~\cite{DBLP:journals/ipl/KhullerMN99}. Generalizing the greedy approach by a simple modification, however, they showed that it could guarantee optimal $(1/2)(1-1/e)$-approximation in particular submodular function for maximum coverage problem. The generalized greedy algorithm (Algorithm~\ref{alg:GGA}) compares the result of the iteratively found solution with the best feasible single-element solution and returns the maximum. Krause and Guestrin extended $(1/2)(1-1/e)$ approximation ratio to maximizing submodular functions under linear constraint~\cite{krause2005note}, which had been improved to $(1-1/\sqrt{e})$ by Lin and Bilmes~\cite{DBLP:conf/naacl/LinB10}.

In the case of monotone submodular functions, greedy algorithms are often able to achieve the best possible worst case approximation guarantee (unless P=NP). Motivated by many real-world applications, the performance of greedy algorithms in general cost functions have also been considered later. Iyer and Bilmes studied cases that the cost function is monotone and submodular~\cite{DBLP:conf/nips/IyerB13}, while Zhang and Vorobeychik investigated problems with only monotone cost functions~\cite{DBLP:conf/aaai/ZhangV16} as a more general case.

Recently, Pareto optimization approaches have been investigated for a wide range of subset selection problems~\cite{DBLP:journals/ec/FriedrichN15,DBLP:conf/nips/QianYZ15,DBLP:journals/ai/QianYTYZ19,DBLP:conf/ijcai/QianSYT17,qian2020multi}. 
It has been shown in~\cite{DBLP:conf/ijcai/QianSYT17} that an algorithm called POMC is able to achieve a $\phi= (\alpha_f/2)(1-\frac{1}{e^{\alpha_f}})$-approximation where $\alpha_f$ measures the closeness of the considered function $f$ to submodularity. The approximation matches the worst-case performance ratio of the generalized greedy algorithm~\cite{DBLP:conf/aaai/ZhangV16}. 

In this paper, we study monotone functions with a dynamic constraint where the constraint bound $B$ changes over time. Such constraint changes reflect real-world scenarios where resources vary during the process. We show that greedy algorithms have difficulties in adapting their current solutions after changes have happened. In particular, we show that there are simple dynamic versions of the classical knapsack problem where adding elements in a greedy fashion when the constraint bound increases over time can lead to an arbitrary bad performance ratio. For the case where constraint bounds decrease over time, we introduce a submodular graph covering problem and show that the considered adaptive generalized greedy algorithm may encounter an arbitrarily bad performance ratio on this problem.

Investigating POMC, we theoretically prove that this algorithm obtains for each constraint bound $b \in [0,B]$, a $\phi= (\alpha_f/2)(1-\frac{1}{e^{\alpha_f}})$-approximation efficiently. Furthermore, when relaxing the bound $B$ to $B^*>B$, $\phi$\nobreakdash-approximations for all values of $b \in [0, B^*]$ are obtained efficiently. The theoretical results for the performance of adaptive greedy algorithm and POMC, in addition to an initial experimental investigation on the quality of solutions found by these algorithms have been published in the proceedings of the Thirty-Third AAAI Conference on Artificial Intelligence (AAAI 2019)~\cite{DBLP:conf/aaai/RoostapourN0019}. 

We now describe how this article extends its conference version. We significantly expand our experiments by considering two real-world problems under two different constraint types and thirty dynamic benchmarks for each of them. We create a baseline by performing our best algorithms for one million generations and statistically compare the results based on offline errors. Moreover, we compare the performance of our algorithms with two additional algorithmic approaches, namely EAMC~\cite{bianefficient} and NSGA-II~\cite{DBLP:journals/tec/DebAPM02}. EAMC is the most recent in the literature that guarantees a polynomial expected time and NSGA-II is a well-known evolutionary multi-objective algorithm used in many applications.

In the first part of our extension, by benchmarking the generalized greedy algorithm, its adaptive version and POMC on the influence maximization problem in social networks over sequences of dynamic changes, we show that POMC obtains superior results to the generalized greedy and adaptive generalized greedy algorithms, specifically when the cost function depends on the structure of the graph. We consider the experimental results in four intervals to examine whether the algorithms behave differently during the optimization process. Here, POMC shows an even more prominent advantage as more dynamic changes are carried out because it is able to improve the quality of its solutions and cater for changes that occur in the future.

To compare these algorithms with EAMC and NSGA-II as the second part of the extension, we study the maximum coverage problem using two graph benchmarks with different sizes and two cost functions. For the ``random'' cost, we assign a random value to each node of the graph as their cost value. For the other cost function, called ``outdegree'', the cost is calculated based on the outdegree of each node. Presenting some explanatory examples, we describe the major problems that hold EAMC back from tracking optimal solutions in dynamic environments. Our results also show that EAMC, which limits its population size to guarantee the polynomial expected time, is not capable of dealing with dynamic changes and performs worse than the other algorithms. We also consider NSGA-II with two population sizes NSGA-II20 population size of 20 and NSGA-II100 with population size of 100. This algorithm that is known to be a high performing evolutionary multi-objective algorithm,  beats POMC when the constraint costs are random. However, NSGA-II20, because of its limited population size and also the effect of the crowding-distance factor in the selection procedure which forces the population to be distributed along the Pareto front, is inferior to POMC for the outdegree cost function. NSGA-II100, on the other hand, is more powerful when changes happen with medium or low frequencies. The effect of larger population size is seen in the experiments which shows NSGA-II20 tracks the optimal solution in high frequent dynamic environment better than NSGA-II100.

The paper is structured as follows. In the next section, we formulate the dynamic problems. In Section~\ref{sec:Theory}, we firstly introduce the algorithms that are studied theoretically and then show that the adaptive generalized greedy algorithm may encounter an arbitrary bad performance ratio even when starting with an optimal solution for the initial budget. In contrast, we show that POMC can maintain $\phi$\nobreakdash-approximation efficiently. Section~\ref{sec:Experiments} presents our experimental setting, the results, and the empirical analysis. We discuss our experimental investigations for influence maximization in social networks in Section~\ref{sec:InfMax}. Then we extend our investigations by considering NSGA-II20, NSGA-II100, and EAMC facing the benchmarks for maximum coverage problem in Section~\ref{sec:MaxCov}, followed by experimental demonstration on how the population size of POMC algorithm changes during one dynamic benchmark. Finally, we finish with some concluding remarks.

\section{Problem Formulation}
\label{sec2}

In this paper we consider optimization problems in which the cost function and objective functions are monotone and quantified according to their closeness to submodularity. 

Different definitions are given for submodularity~\cite{DBLP:journals/mp/NemhauserWF78} and we use the following one in this paper. For a given set $V=\{v_1,\cdots,v_n\}$, a function $f: 2^V\rightarrow\mathbb{R}$ is submodular if for any $X\subseteq Y\subseteq V$ and $v\notin Y$
\begin{align}\label{equ:submodular}
f(X\cup \{v\}) - f(X)\geq f(Y\cup \{v\})-f(Y)\text{.}
\end{align}

In addition, we consider how much a function is close to being submodular, measured by the submodularity ratio~\cite{DBLP:conf/aaai/ZhangV16}.
The function $f$ is $\alpha_f$-submodular where
\begin{align}\label{equ:submodularity}\alpha_f = \min_{X\subseteq Y,v\notin Y}\frac{f(X\cup \{v\}) - f(X)}{f(Y\cup \{v\})- f(Y)}\text{.}\end{align}
This definition is equivalent to the Equation~\ref{equ:submodular} when ${\alpha_f=1}$, i.e., we have submodularity in this case. Another notion which is used in our analysis is the curvature of function $f$. The curvature measures the deviation from linearity and reflects the effect of marginal contribution according to the function $f$~\cite{DBLP:journals/dam/ConfortiC84,vondrak2010submodularity}. For a monotone submodular function $f:2^V \rightarrow \mathbb{R}^+$,
$$\kappa_f=1-\min_{v\in V}\frac{f(V)-f(V\setminus \{v\})}{f(v)}$$
is defined as the total curvature of $f$.

In many applications the function to be optimized, $f$, comes with a cost function $c$ which is subject to a given cost constraint $B$. Often the cost function $c$ cannot be evaluated precisely. Hence, function $\hat{c}$ is used which is an $\alpha_c$-submodular $\psi$-approximation of cost $c$~\cite{DBLP:conf/aaai/ZhangV16}. To be more precise, according to submodularity definition (Equation \ref{equ:submodularity}), $\hat{c}$ is chosen such that $\alpha_c = \alpha_{\hat{c}}$ and   $c(X)\leq\hat{c}(X)\leq\psi c(X)$, where $n=|V|$. Moreover, according to the submodularity of $f$, the aim is to find a good approximation instead of finding an optimal solution.

Consider the static version of an optimization problem defined in~\cite{DBLP:conf/ijcai/QianSYT17}.
\begin{definition}[The Static Problem] \label{def:static version} Given a monotone objective function $f:2^V\rightarrow \mathbb{R}^+$, a monotone cost function $c:2^V\rightarrow\mathbb{R}^+$ and a budget $B$, the goal is to compute a solution $X$ such that
	$$X = \argmax_{Y\subseteq V}f(Y) \text{ s.t. } c(Y)\leq B\text{.} $$
\end{definition}

Following the investigation of the static case in~\cite{DBLP:conf/ijcai/QianSYT17}, we are interested in a $\phi$\nobreakdash-approximation where
$\phi=(\alpha_f/2)(1-\frac{1}{e^{\alpha_f}})$ depends on the submodularity ratio.

Zhang and Vorobeychik considered the performance of the generalized greedy algorithm~\cite{DBLP:conf/aaai/ZhangV16}, given in Algorithm~\ref{alg:GGA}, according to the approximated cost function $\hat{c}$. Starting from the empty set, the algorithm always adds the element with the largest objective to cost ratio that does not violate the given constraint bound $B$. 

\begin{algorithm}[t]
	\SetKwInOut{Input}{input}
	\Input{Initial
		budget constraint $B$.}
	$X\leftarrow\emptyset$\;
	$V^\prime \leftarrow V$\;
	\Repeat{$V^\prime \leftarrow \emptyset$}{$v^*\leftarrow \argmax_{v\in V^\prime}\frac{f(X\cup \{v\})-f(X)}{\hat{c}(X\cup \{v\})-\hat{c}(X)}$\;
		\If{$\hat{c}(X\cup \{v^*\})\leq B$}{$X\leftarrow X\cup \{v^*$\}\;}
		$V^\prime \leftarrow V^\prime\setminus \{v^*\}$\;}
	$v^* \leftarrow \argmax_{v\in V;\hat{c}(v)\leq B}f(v)$\;
	\Return{$\argmax_{S\in \{X,v^*\}}f(S)$}\

	\caption{Generalized Greedy Algorithm}\label{alg:GGA}
\end{algorithm}

Let $K_c = \max\{|X| : c(X)\leq B\}$. The optimal solution $\tilde{X}_B$ in these investigations is defined as $\tilde{X}_B = \arg \max\{f(X)\mid c(X)\leq\alpha_c \frac{B(1+\alpha_c^2(K_c-1)(1-\kappa_c)))}{\psi K_c}\}$ where $\alpha_c$ is the submodularity ratio of $c$. This formulation gives the value of an optimal solution for a slightly smaller budget constraint. The goal is to obtain a good approximation of $f(\tilde{X}_B)$ in this case.

It has been shown in~\cite{DBLP:conf/aaai/ZhangV16} that the generalized greedy algorithm, which adds the item with the highest marginal contribution to the current solution in each step, achieves a $(1/2)(1-\frac{1}{e})$-approximate solution if $f$ is monotone and submodular.~\cite{DBLP:conf/ijcai/QianSYT17} extended these results to objective functions with $\alpha_f$ submodularity ratio and proved that the generalized greedy algorithm obtains a $\phi=(\alpha_f/2)(1-\frac{1}{e^{\alpha_f}})$-approximation.
For the remainder of this paper, we assume $\phi=(\alpha_f/2)(1-\frac{1}{e^{\alpha_f}})$ and are interested in obtaining solutions that are $\phi$\nobreakdash-approximation for the considered problems.

In this paper, we study the dynamic version of problem given in Definition~\ref{def:static version}.

\begin{definition} [The Dynamic Problem]
	\label{defdyn}
	Let $X$ be a $\phi$\nobreakdash-approximation for the problem in Definition~\ref{def:static version}. The dynamic problem is given by a sequence of changes where in each change the current budget $B$ changes to $B^*=B+d$, $d \in \mathds{R}_{\geq - B}$. The goal is to compute a $\phi$\nobreakdash-approximation $X^\prime$ for each newly given budget $B^*$.
\end{definition}

The Dynamic Problem evolves over time by the changing budget constraint bounds. Note that every fixed constraint bound gives a static problem and any good approximation algorithm can be run from scratch for the newly given budget. However, the main focus of this paper are algorithms that can adapt to changes of the constraint bound.

\section{Theoretical Analysis}\label{sec:Theory}
POMC and GGA are proven to find a $\phi$\nobreakdash-approximated solution on the static version of submodular subset selection problems. In this section, we analyze their performance in a dynamic environment. We first consider an extended version of GGA, which is an adaptive version to perform in the dynamic environment. Then we prove the power of POMC in computing $\phi$\nobreakdash-approximation in static and dynamic versions.

\subsection{Algorithms}\label{Sec:Algs}

We consider dynamic problems according to Definition~\ref{defdyn} with $\phi= (\alpha_f/2)(1-\frac{1}{e^{\alpha_f}})$ and are interested in algorithms that adapt their solutions to the new constraint bound $B^*$ and obtain a $\phi$\nobreakdash-approximation for the new bound $B^*$. As the generalized greedy algorithm can be applied to any bound $B$, the first approach would be to run it for the newly given bound $B^*$. However, this might lead to totally different solutions and adaptation of already obtained solutions might be more beneficial. Furthermore, adaptive approaches that change the solution based on the constraint changes are of interest as they might be faster in obtaining such solutions and/or be able to learn good solutions for the different constraint bounds that occur over time.

\begin{algorithm}[t]
	\SetKwInOut{Input}{input}
	Let $B^*$ be the new budget\;
	\uIf {$B^*<B$}{
		\While {$\hat{c}(X)> B^*$}{
			$v^*\leftarrow \argmin_{v\in X} \frac{ f(X)-f(X\setminus\{v\})}{\hat{c}(X)-\hat{c}(X\setminus \{v\})}$\;
			$X \leftarrow X\setminus \{v^*\}$ \;}
	}
	\ElseIf{$B^*> B$}{
		$V^\prime \leftarrow V\setminus X$\;
		\Repeat{$V^\prime \leftarrow \emptyset$}{
			$v^*\leftarrow \argmax_{v\in V^\prime}\frac{f(X\cup \{v\})-f(X)}{\hat{c}(X\cup \{v\})-\hat{c}(X)}$\;
			\If{$\hat{c}(X\cup \{v^*\})\leq B^*$}{$X\leftarrow X\cup \{v^*$\}\;}
			$V^\prime \leftarrow V^\prime\setminus \{v^*\}$\;}
		
	}
	$v^* \leftarrow \argmax_{v\in V;\hat{c}(v)\leq B^*}f(v)$ \label{alg_line:vstar}\;
	\Return{$\argmax_{S\in \{X,v^*\}}f(S)$}\
	
	\caption{Adaptive Generalized Greedy Algorithm}\label{alg:AGA}
	
\end{algorithm}

Based on the generalized greedy algorithm, we introduce the adaptive generalized greedy algorithm (Algorithm \ref{alg:AGA}). This algorithm is modified from Algorithm~\ref{alg:GGA} in a way that enables it to deal with a dynamic change. Let $X$ be the current solution of the algorithm. When a dynamic change decreases the budget constraint, the algorithm removes items from $X$ according to their marginal contribution, until it achieves a feasible solution. When there is a dynamic increase, this algorithm behaves similarly to the generalized greedy algorithm.

Furthermore, we consider the Pareto optimization approach POMC (Algorithm~\ref{alg:POMC}) which is also known as Global SEMO in the evolutionary computation literature~\cite{DBLP:journals/tec/LaumannsTZ04,DBLP:journals/ec/FriedrichHHNW10,DBLP:journals/ec/FriedrichN15}. POMC is a multi-objective optimization approach which is proven to perform better than the generalized greedy algorithm in case of local optima~\cite{DBLP:conf/ijcai/QianSYT17}. We reformulate the problem as a bi-objective problem in order to use POMC as follows:
\begin{flalign}\nonumber
&\hspace{1.5cm}\argmax_{X\in \{0,1\}^n} (f_1(X),f_2(X))\text{,}\\\nonumber
&\text{where:  }\\\nonumber
& f_1(X)= 
\begin{cases}
-\infty,  &  \hat{c}(X) > B+1 \\
f(X),  & \text{otherwise}
\end{cases}
,f_2(X)=-\hat{c}(X)\text{.}
\end{flalign}

\begin{algorithm}[t]
	\SetKwInOut{Input}{input}
	\Input{Initial
		budget constraint $B$, time $T$}
	$X\leftarrow \{0\}^n$\;
	Compute $(f_1(X), f_2(X))$\;
	$P\leftarrow \{x\}$\;
	$t\leftarrow 0$\;
	\While{$t<T$}{
		Select $X$ from $P$ uniformly at random\;
		$X^\prime \leftarrow$ flip each bit of $X$ with probability $\frac{1}{n}$\;
		Compute $(f_1(X^\prime), f_2(X^\prime))$\;
		\If{$\nexists Z\in P$ such that $Z\succ X^\prime$}{
			$P\leftarrow(P\setminus \{Z\in P\mid X^\prime\succeq Z\})\cup\{X^\prime\}$\;
		}
		$t=t+1$\;
	}
	\Return{$\argmax_{X\in P:\hat{c}(X)\leq B}f(x)$}

	\caption{POMC Algorithm}\label{alg:POMC}
\end{algorithm}

This algorithm optimizes the cost function and the objective function simultaneously. To achieve this, it uses the concept of dominance to compare two solutions. Solution $X_1$ dominates $X_2$, denoted by $X_1\succeq X_2$, if $f_1(X_1)\geq f_1(X_2) \land f_2(X_1)\geq f_2(X_2)$. The dominance is strict, $\succ$, when at least one of the inequalities is strict. POMC produces a population of non-dominated solutions and optimizes them during the optimization process. In each iteration, it chooses a solution $X$ randomly from the population and flips each bit of the solution with the probability of $1/n$. It adds the mutated solution $X^\prime$ to the population only if there is no solution in the population that strictly dominates $X^\prime$. All the solutions which are dominated by $X^\prime$ will be deleted from the population afterward.

Note that we only compute the objective vector $(f_1(X), f_2(X))$ when the solution $X$ is created. This implies that the objective vector is not updated after changes to the constraint bound $B$. As a consequence, solutions whose constraint exceeds the value of $B+1$ for a newly given bound are kept in the population. However, newly produced individuals exceeding $B+1$ for the current bound $B$ are not included in the population as they are dominated by the initial search point $0^n$. We are using the value $B+1$ instead of $B$ in the definition of $f_1$ as this gives the algorithm some look ahead for larger constraint bounds. However, every value of at least $B$ would work for our theoretical analyses. The only drawback would be a potentially larger population size which influences the value $P_{\max}$ in our runtime bounds.

\subsection{Adaptive Generalized Greedy Algorithm}\label{sec:AdGGA}
In this section we analyze the performance of the adaptive generalized greedy algorithm. This algorithm is a modified version of the generalized greedy using the same principle in adding and deleting items. However, in this section we prove that the adaptive generalized greedy algorithm is not able to deal with the dynamic change, i.e., the approximation obtained can become arbitrarily bad during a sequence of dynamic changes.

In order to show that the adaptive generalized greedy algorithm can not deal with dynamic increases of the constraint bound, we consider a special instance of the classical knapsack problem. Note that the knapsack problem is special submodular problem where both the objective and the cost function are linear.

Given $n+1$ items $e_i=(c_i,f_i)$ with cost $c_i$ and value $f_i$ independent of the choice of the other items, we assume there are items $e_i= (1,\frac{1}{n})$, $1\leq i\leq n/2$, $e_i= (2,1)$, $n/2+1\leq i\leq n$, and a special item $e_{n+1} = (1,3)$. We have $f_{inc}(X) = \sum_{e_i \in X} f_i$ and $c_{inc}(X) = \sum_{e_i \in X} c_i$ as the linear objective and constraint function, respectively.

\begin{theorem}
	Given the dynamic knapsack problem $(f_{inc}, c_{inc})$, starting with $B=1$ and increasing the bound $n/2$ times by $1$, the adaptive generalized greedy algorithm computes a solution that has approximation ratio $O(1/n)$.
\end{theorem}

\begin{proof}
	For the constraint $B=1$ the optimal solution is $\{e_{n+1}\}$. Now let there be $n/2$ dynamic changes where each of them increases $B$ by $1$. In each change, the algorithm can only pick an item from $\{e_1,\cdots,e_{n/2}\}$, otherwise it violates the budget constraint.
	After $n/2$ changes, the budget constraint is $1+n/2$ and the result of the algorithm is $S=\{e_{n+1},e_1,\cdots,e_n/2\}$ with $f(S)=3 +(n/2) \cdot (1/n)= 7/2$ and $c(S) = 1+n/2$. However, an optimal solution for budget $1+n/2$ is $S^* = \{e_{n+1},e_{n/2+1},\ldots, e_{\frac{3n}{4}}\}$ with $f(S^*)=3+\frac{n}{4}$. Hence, the approximation ratio in this example is 
	$(7/2)/(3+n/4) = O(1/n)$. 
\end{proof}

Now we consider the case where the constraint bound decreases over time and show that the adaptive generalized greedy algorithm may also encounter situations where the approximation ratio becomes arbitrarily bad over time.

We consider the following \emph{Graph Coverage Problem}. Let $G=(U,V,E)$ be a bipartite graph with bipartition $U$ and $V$ of vertices with $|U|=n$ and $|V|=m$. The goal is to select a subset $S \subseteq U$ with $|S| \leq B$ such that the number of neighbors of $S$ in $V$ is maximized. Note that the objective function $f(S)$ measuring the number of neighbors of $S$ in $V$ is monotone and submodular.

\begin{figure}
	\centering
	\includegraphics[width=.44\textwidth]{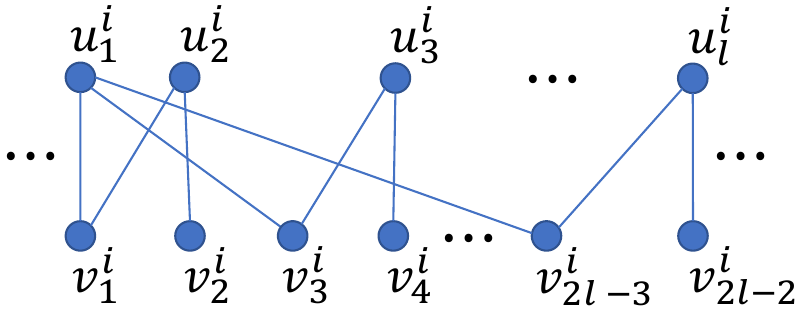}
	\caption{Single subgraph $G_i$ of $G = (U,V,E)$}
	\label{fig:example}
\end{figure}
We consider the graph $G = (U,V,E)$ which consists of $k$ disjoint subgraphs 
$$G_i=(U_i=\{u_1^i,\cdots,u_l^i\},V_i=\{v_1^i,\cdots,v_{2l-2}^i\},E_i)$$
(see Figure~\ref{fig:example}).
Node $u_1^i$ is connected to nodes $v_{2j-1}^i$, $1\leq j\leq l-1$. Moreover, each vertex $u_j^i$, $2 \leq j\leq l$ is connected to two vertices $v_{2j-3}^i$ and $v_{2j-2}^i$. We assume that $k = \sqrt{n}$ and $l= n/k =\sqrt{n}$.

\begin{theorem} \label{thm:adapt_dyn}
	Starting with the optimal set $S=U$ and budget $B=n$, there is a specific sequence of dynamic budget reductions such that the solution obtained by the adaptive generalized greedy algorithm
	has the approximation ratio $O(1/\sqrt{n})$.
\end{theorem}

\begin{proof}

	Let the adaptive generalized greedy algorithm be initialized with $X = U$ and $B = n=kl$. We assume that the budget decreases from $n$ to $k$ where each single decrease reduces the budget by $1$.
	In the first $k$ steps, to change the cost of solution from $n$ to $n-k$, the algorithm removes the nodes $u_1^i$, $1\leq i\leq k$, as they have a marginal contribution of $0$. Following these steps, all the remaining nodes have the same marginal contribution of $2$. The solution $X$ of size $k$ obtained by the removal steps of the adaptive generalized greedy algorithm contains $k$ vertices which are connected to $2k$ nodes of $V$, thus $f(X)=2k=2\sqrt{n}$. Such a solution is returned by the algorithm for $B=k$ as the most valuable single node has at most $(l-1)=(\sqrt{n}-1)$ neighbors in $V$.
	For $B=k$, the optimal solution $X^* = \{u_1^i \mid 1\leq i\leq k\}$ has $f(X^*) = k(l-1)=n- \sqrt{n}$. Therefore, the approximation ratio achieved by the adaptive generalized greedy algorithm is upper bounded by
	$(2\sqrt{n})/(n-\sqrt{n})= O(1/\sqrt{n}).$
\end{proof}

\subsection{Pareto Optimization}
In this section we analyze the behavior of POMC facing a dynamic change. According to Lemma 3 in~\cite{DBLP:conf/ijcai/QianSYT17}, for any $X\subseteq V$ and $v^*= \argmax_{v\notin X}\frac{f(X\cup \{v\})-f(X)}{\hat{c}(X\cup \{v\})-\hat{c}(X)}$ we have $$f(X\cup \{v^*\})-f(X)\geq\alpha_f\frac{\hat{c}(X\cup \{v^*\})-\hat{c}(X)}{B}\cdot (f(\tilde{X})-f(X))\text{.}$$ 

We denote by
$\delta_{\hat{c}}=\min\{\hat{c}(X\cup \{v\})-\hat{c}(X)\mid X\subseteq V,v\notin X\}$
the smallest contribution of an element to the cost of a solution for the given problem. Moreover, let $P_{\max}$ be the maximum size of POMC's population during the optimization process.

The following theorem considers the static case and shows that POMC computes a $\phi$\nobreakdash-approximation efficiently for every budget $b \in [0,B]$.

\begin{theorem}\label{thm:POMC_running_time}
	Starting from $\{0\}^n$, POMC computes, for any budget $b \in [0,B]$, a $\phi=(\alpha_f/2)(1-{1}/{e^{\alpha_f}})$-approximate solution after $T=cnP_{\max}\cdot \frac{B}{\delta_{\hat{c}}}$ iterations with the constant probability, where $c\geq 8e+1$ is a sufficiently large arbitrary constant.
\end{theorem}
\begin{proof}
	We first consider the number of iterations to find a $(\alpha_f/2)\left(1-(1-\frac{\alpha_f}{k})^k\right)$-approximate solution for a budget $b\in [0,B]$ and some $k$. We consider the largest value of $i$ for which there is a solution $X$ in the population where $\hat{c}(X)\leq i<b$ and $$f(X)\geq \left(1-\left(1-\alpha_f \frac{i}{bk}\right)^k\right)\cdot f({\tilde{X}}_b)$$ holds for some $k$. Initially, it is true for $i=0$ with $X = \{0\}^n$. 
	We now show that adding $v^*$ to the current solution has the desired contribution to achieve a $\phi$\nobreakdash-approximate solution. Let $X\subseteq V$ and $v^*= \argmax_{v\notin X}\frac{f(X\cup \{v\})-f(X)}{\hat{c}(X\cup \{v\})-\hat{c}(X)}$. 
	
	Assume that 
	$$f(X)\geq \left(1-\left(1-\alpha_f \frac{i}{bk}\right)^k\right)\cdot f({\tilde{X}}_b)$$ holds for some $\hat{c}(X)\leq i<b$ and $k$. 
	Then adding $v^*$ leads to
	\begin{align}\nonumber
	&f(X\cup \{v^*\})\geq\\\nonumber
	&\left(1-\left(1-\alpha_f \frac{i+\hat{c}(X\cup \{v^*\})-\hat{c}(X)}{b(k+1)}\right)^{k+1}\right)\cdot f({\tilde{X}}_b).
	\end{align}
	This process only depends on the quality of $X$ and is independent of its structure.
	Starting from $\{0\}^n$, if the algorithm carries out such steps at least $b/\delta_{\hat{c}}$ times, it reaches a solution $X$ such that
	\begin{align}\nonumber
	f(X\cup \{v^*\})&\geq\left(1-\left(1-\alpha_f \frac{b}{bk^*}\right)^{k^*}\right)\cdot f({\tilde{X}}_b)\\\nonumber
	&\geq \left(1-\frac{1}{e^{\alpha_f}}\right)\cdot f({\tilde{X}}_b)\text{.}
	\end{align}
	Considering item $z= \argmax_{v \in V:\hat{c}(v)\leq b}f(v)$, by submodularity and $\alpha_f\in[0,1]$ we have
	$
	f(X\cup \{v^*\}) \leq (f(X)+f(z))/\alpha_f.
	$

	This implies that 
	$$\max\{f(X),f(z)\}\geq(\alpha_f/2)\cdot(1-\frac{1}{e^{\alpha_f}})\cdot f({\tilde{X}}_b)\text{.}$$
	
	We consider $T=cnP_{\max} B/\delta_{\hat{c}}$ iterations of the algorithm and analyze the success probability within $T$ steps. To have a successful mutation step where $v^*$ is added to the currently best approximate solution, the algorithm has to choose the right individual in the population, which happens with probability at least $1/P_{\max}$. Furthermore, the single bit corresponding to $v^*$ has to be flipped which has the probability at least $1/(en)$.
	We call such a step a success. 
	Let random variable $Y_j = 1$ when there is a success in iteration $j$ of the algorithm and $Y_j = 0$, otherwise. Thus, we have 
	$$\Pr(Y_j=1)\geq\frac{1}{en}\cdot\frac{1}{P_{\max}}$$
	as long as a $\phi$\nobreakdash-approximation for bound $b$ has not been obtained.
	
	Furthermore, let $Y_i^*$, $1 \leq i \leq T$, be mutually independent random binary variables with $\Pr[Y^*_i=1] = \frac{1}{enP_{\max}}$ and $\Pr[Y^*_i=0] = 1-\frac{1}{enP_{\max}}$. For the expected value of the random variable $Y^*=\sum_{j=1}^{T}Y^*_j$ we have
	$$E[Y^*]=\frac{T}{enP_{\max}} = \frac{cB}{e \delta_{\hat{c}}} \geq \frac{cb}{e \delta_{\hat{c}}}\text{.}$$
	We use Lemma 1 in~\cite{DBLP:journals/ec/DoerrHK11} for moderately correlated variables which allows the use of the following Chernoff bound
	\begin{align}\label{equ:chernoff}\nonumber
	\Pr\left(Y<(1-\delta)E[Y^*]\right)&\leq \Pr\left(Y^*<(1-\delta)E[Y^*]\right)\\
	&\leq e^{-E[Y^*]\delta^2/2}\text{.}
	\end{align}
	
	Using Equation~\ref{equ:chernoff} with $\delta = (1-\frac{e}{c})$, we bound the probability of not finding a $\phi$\nobreakdash-approximation of ${\tilde{X}}_b$ in time $T=cnP_{\max}{B}/{\delta_{\hat{c}}}$ by
	\begin{align}\nonumber
	\Pr(Y\leq \frac{b}{\delta_{\hat{c}}})&\leq e^{-\frac{(c-e)^2 B}{2ce\delta_{\hat{c}}}} \leq e^{-\frac{(c/2)^2 B}{2ce\delta_{\hat{c}}}}\\\nonumber
	&\leq e^{-\frac{cB}{8e\delta_{\hat{c}}}}\leq e^{-\frac{B}{\delta_{\hat{c}}}}\text{.}
	\end{align}
	
	Using the union bound and taking into account that there are at most $B/\delta_{\hat{c}}$ different values for $b$ to consider, the probability that there is a $b \in [0,B]$ for which no $\phi$\nobreakdash-approximation has been obtained is upper bounded by $\frac{B}{\delta_{\hat{c}}}\cdot e^{-\frac{B}{\delta_{\hat{c}}}}$.
	
	This implies that POMC finds a $(\alpha_f/2)(1-\frac{1}{e^{\alpha_f}})$-approximate solution with probability at least
	$1-\frac{B}{\delta_{\hat{c}}} \cdot e^{-\frac{B}{\delta_{\hat{c}}}}$ for each $b \in [0,B]$. 
\end{proof}
Note that if we have $B/\delta_{\hat{c}}\geq\log{n}$ then the probability of achieving a $\phi$\nobreakdash-approximation for every $b \in [0,B]$ is $1-o(1)$. In order to achieve a probability of $1-o(1)$ for any possible change, we can run the algorithm for $T'=cnP_{\max}\cdot \max\{\frac{B}{\delta_{\hat{c}}}, \log n\}$, $c\geq 8e+1$, iterations.

Now we consider the performance of POMC in the dynamic version of the problem. In this version, it is assumed that POMC has achieved a population which includes a $\phi$\nobreakdash-approximation for all budgets $b\in [0,B]$. Reducing the budget from $B$ to $B^*$ implies that a $\phi$\nobreakdash-approximation for the newly given budget $B^*$ is already contained in the population.

Consideration must be given to the case where the budget increases.
Assume that the budget changes from $B$ to $B^*=B+d$ where $d>0$. We analyze the time until POMC has updated its population such that it contains for any $b \in [0,B^*]$ a $\phi$\nobreakdash-approximate solution.

We define 
\begin{flalign*}\nonumber
I_{\max}(b,b^\prime)&=\max\{ i\in[0,b]\mid \exists X\in P,\hat{c}(X)\leq i \\\nonumber
&\land f(X)\geq \left(1-\left(1-\alpha_f \frac{i}{bk}\right)^k\right)\cdot f({\tilde{X}}_b)\\\nonumber
&\land f(X)\geq \left(1-\left(1-\alpha_f \frac{i}{b^\prime k^\prime}\right)^{k^\prime}\right)\cdot f({\tilde{X}}_{b^\prime}) \}\nonumber
\end{flalign*}
for some $k$ and $k^\prime$. The notion of $I_{\max}(b,b^\prime)$ enables us to correlate progress in terms of obtaining a $\phi$\nobreakdash-approximation for budgets $b$ and $b^\prime$.
When increasing the budget from $B$ to $B^*$, it allows us to define the solution in the current population $P$ to which we can add an element in order to obtain good approximations for the budgets $b' \in (B,B^*]$.

\begin{theorem}
	Let POMC have a population $P$ such that, for every budget $b\in [0,B]$, there is a $\phi$\nobreakdash-approximation in $P$.
	After changing the budget to $B^*>B$, POMC has computed a $\phi$\nobreakdash-approximation with probability $\Omega(1)$ within $T = cnP_{\max}\frac{d}{\delta_{\hat{c}}}$ steps for every $b\in [0,B^*]$, where $c\geq 8e+1$ is a sufficiently large arbitrary constant.
	
\end{theorem}

\begin{proof}
	Let $P$ denote the current population of POMC in which, for any budget $b\leq B$, there is a $\left(1-(1-\frac{\alpha_f}{k})^k\right)$-approximate solution for some $k$.
	Such approximations are not lost during the run of the algorithm and we need to show that approximations for budgets $b' \in (B, B^*]$ are also added to the population. We consider an arbitrary budget $b' \in (B,B^*]$ and analyze the time to obtain  $\phi$-approximation for $b'$.
	Let $X$ be the solution corresponding to $I_{\max}(B,b')$, where $b' \in (B,B^*]$.
	Moreover, assume that $v^*= \argmax_{v\notin X}\frac{f(X\cup \{v\})-f(X)}{\hat{c}(X\cup \{v\})-\hat{c}(X)}$ denotes the item with the highest marginal contribution which could be added to $X$ and $X^\prime = X\cup \{v^*\}$. 
	According to Lemma 3 and Theorem 2 in~\cite{DBLP:conf/ijcai/QianSYT17} and the definition of $I_{\max}(B,b')$, we have
	
	
	\begin{flalign*}\nonumber
	&f(X^\prime)\geq \left(1-\left(1-\alpha_f \frac{I_{\max}+\hat{c}(X^\prime)-\hat{c}(X)}{Bk}\right)^k\right)\cdot f({\tilde{X}}_B)\\\nonumber
	&\text{and}\\\nonumber
	&f(X')\geq \left(1-\left(1-\alpha_f \frac{I_{\max}+\hat{c}(X^\prime)-\hat{c}(X)}{b' k'}\right)^{k'}\right)\cdot f({\tilde{X}}_{b'})\text{.}
	\end{flalign*}
	This implies that adding $v^*$ to $X$ violates the budget constraint $B$, otherwise we would have a greater value for $I_{\max}$.
	
	If $I_{\max}+\hat{c}(X^\prime)-\hat{c}(X)\geq b'$, then, similar to the proof of Theorem~\ref{thm:POMC_running_time}, we have 
	$$\max\{f(X),f(z)\}\geq(\alpha_f/2)\cdot \left(1-\frac{1}{e^{\alpha_f}}\right)\cdot f(\tilde{X}_{b'}).$$
	
	Otherwise, we have
	$$f(X^\prime)\geq \left(1-\left(1-\alpha_f \frac{B}{b' k^\prime} \right)^{k^\prime}\right)\cdot f({\tilde{X}}_{b'})\text{.}$$
	From this point, the argument in the proof of Theorem~\ref{thm:POMC_running_time} holds and POMC obtains $b'$, a $\phi$\nobreakdash-approximation after at most $\frac{d}{\delta_{\hat{c}}}$ successes. 
	Using the Chernoff bound and the union bound as done in Theorem 5, the population contains after at most $T=cnP_{\max}d/\delta_{\hat{c}}$ iterations for each $b'\in (B,B^*]$ also a $\phi=(\alpha_f/2)(1-\frac{1}{e^{\alpha_f}})$-approximation with probability at least 
	$1- \frac{d}{\delta_{\hat{c}}} \cdot e^{-\frac{d}{\delta_{\hat{c}}}}$. 
\end{proof}

Note that if the dynamic change is sufficiently large such that $\frac{d}{\delta_{\hat{c}}}\geq \log{n}$, then the probability of having obtained a $\phi$\nobreakdash-approximation, for every budget $b \in [0,B^*]$, increases to $1-o(1)$. A success probability of $1-o(1)$ can be obtained for this magnitude of changes by giving the algorithm time $T'= cnP_{\max} \max\{\frac{d}{\delta_{\hat{c}}}, \log n\}$, where $c\geq 8e+1$.

A special class of known submodular problems is the maximization of a function with a cardinality constraint. In this case, the constraint value can take on at most $n+1$ different values and we have $P_{\max} \leq n+1$. Furthermore, we have $\delta=1$ which leads to the following two corollaries.
\begin{corollary} Consider the static problem with a cardinality constraint and constraint bound $B$.
	POMC computes, for every budget $b\in[0,B]$, a $\phi$\nobreakdash-approximation within $T=cn^2\cdot\max\{B,\log{n}\}$, $c\geq 8e+1$, iterations with probability $1-o(1)$.
\end{corollary}
\begin{corollary}
	Consider the dynamic problem with a cardinality constraint and constraint bound $B$. Assume that $P$ contains a $\phi$\nobreakdash-approximation for every $b\in [0,B]$. Then after increasing the budget to $B^*$, POMC computes, for every budget $b\in[0,B^*]$, a $\phi$\nobreakdash-approximation in time $T = cn^2\max\{d,\log{n}\}$, $c\geq8e+1$ and $d=|B^*-B|$, with probability $1-o(1)$.
\end{corollary}

\section{Experimental Investigations}\label{sec:Experiments}

In this section, we experimentally compare the performance of our algorithms on dynamic variants of problems, where the constraint bound changes over time. We first analyze the practical performance of three algorithms that have been investigated theoretically in the previous section, the generalized greedy algorithm (GGA), the adaptive generalized greedy algorithm (AdGGA) and POMC, on the dynamic submodular influence maximization problem~\cite{DBLP:conf/aaai/ZhangV16,DBLP:conf/ijcai/QianSYT17}. In addition to the plain POMC described in Algorithm~\ref{alg:POMC}, we also consider another version in which the algorithm has a warm-up phase. Starting from a zero solution, \POMCwp~performs ten thousand generations before the first change happens, which gives it more time to build a population that is prepared for the following changes.

Afterward, we include two more algorithms. EAMC, introduced in~\cite{bianefficient}, is a newer algorithm that theoretically guarantees $\phi$\nobreakdash-approximation in polynomial expected time $2en^2(n + 1)$. The other algorithm is a version of NSGA-II that benefits from a specific elitism to keep track of the best-found solution in a dynamic environment. We compare these two algorithms with \POMCwp, GGA and AdGGA on the dynamic submodular maximum coverage problem.
The following section describes our experimental setting.

\subsection{Experimental Setting}\label{subsec:exp_setting}
\begin{figure}[t]
	\centering
	\includegraphics[width=.5\textwidth]{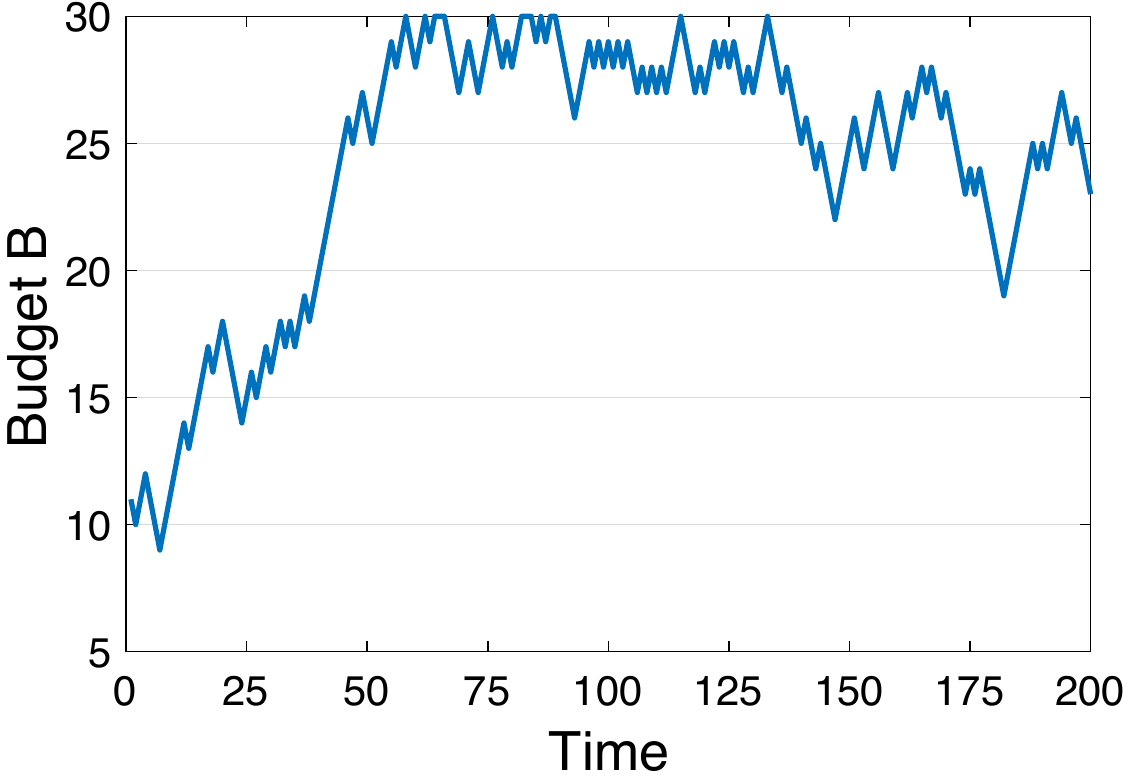}
	\caption{Budget over time for dynamic problems}
	\label{fig:plot}
\end{figure}
We build our dynamic benchmark based on the approach in~\cite{roostapour2020pareto}. We have two problems, and for each problem, we consider two different cost functions. Thus, we study four types of instances. Each instance has an initial, a maximum and a minimum budget denoted by $B_{\text{init}}$, $B_{\max}$ and $B_{\min}$, respectively.
Every $\tau$ evaluation, a dynamic change adds the value of $\delta\in [-r,r]$, which is chosen uniformly at random, to the current budget $B\in [B_{\max},B_{\min}]$. In other words, the iterative algorithms have $\tau$ evaluations to find a solution for budget constraint $B$, before the next change happens. We consider two hundred changes for each run, and there are thirty different sequences of random changes produced for each instance. An example of how the budget values change over time with $B_{\text{init}} = 10$ and $\delta \in \{-1,1\}$ is shown in Figure~\ref{fig:plot}. For the baseline in the influence maximization and the maximum coverage problems, we run \POMCwp~and NSGA-II for all occurred budget constraints for one million generations, respectively.

Algorithms start with the initial budget constraint $B_{\text{init}}$. To calculate the error for each algorithm, let $s^i$ denote the best-found solution right before the change $i$ happens, and $s_{\mathrm{b}}^i$ be the solution found by the baseline algorithm. For each dynamic change, we record error $e_i$ as : $e_i = f(s_{\mathrm{b}}^i) - f(s^i)$. Then for each interval that contains $m$ changes, the partial offline error is $\sum_{i=1}^{m}{e_i/m}$.

We compare the performance of algorithms for each instance according to the thirty partial offline error values. To establish a statistical comparison of the results among different algorithms, we use a multiple comparisons test. In particularity, we focus on the method that compares a set of algorithms. 
For statistical validation, we use the Kruskal-Wallis test with $95$\% confidence. Afterwards, we apply the Bonferroni post-hoc statistical procedures that are used for multiple comparisons of a control algorithm against two or more other algorithms.
For more detailed descriptions of the statistical tests, we refer the reader to~\cite{Corder09}

Our results are summarized in the Tables~\ref{tbl:IM},~\ref{tbl:MC-GGA}, and~\ref{tbl:MC-EAMC}. The columns represent the algorithms with the corresponding mean value and standard deviation.
Note, $X^{(+)}$ is equivalent to the statement that the algorithm in the column outperformed algorithm $X$, and $X^{(-)}$ is equal to the statement that X outperformed the algorithm in the given column. If the algorithm X does not appear, this means that no significant difference was observed between the algorithms.

\subsection{The Influence Maximization Problem}
\label{sec:InfMax}

The influence maximization problem aims to identify a set of most influential users in a social network. Given a directed graph $G=(V,E)$ where each node represents a user. Each edge $(u,v) \in E$ has assigned an edge probability $p_{u,v} ((u,v) \in{E})$. The probability $p_{u,v}$ corresponds to the strengths of influence from user $u$ to user $v$. The goal is to find a subset $X \subseteq V$ such that the expected number of activated nodes from $X$, $IC(X)$, is maximized. Given a cost function $c$ and a budget $B$ the submodular optimization problem is formulated as $\argmax_{X\subseteq V} E[|IC(X)|] \text{ s.t. } c(X)\leq B.$ To calculate $E[|IC(X)|]$ in our experiments, we simulate the random process of influence of solution $X$ for 500 times independently, and use the average as the value of $E[|IC(X)|]$.

We consider two types of cost functions.
The routing constraint takes into account the costs of visiting nodes whereas the cardinality constraint counts the number of chosen nodes. 
For both cost functions, the constraint is met if the cost is at most $B$. 
For more detailed descriptions of the influence maximization through a social network problem, we refer the reader to~\cite{DBLP:conf/aaai/ZhangV16,DBLP:conf/ijcai/QianSYT17,DBLP:journals/toc/KempeKT15}.

To build the dynamic benchmark, as described in Section~\ref{subsec:exp_setting}, we assume that the initial constraint bound is $B_{\text{init}}=10$, which stays within the interval $[5,30]$. $\delta\in\{-1,1\}$ is chosen uniformly at random and we consider four values for $\tau\in\{100,1000,5000,10000\}$. In \POMCwp, we consider the option of POMC having a warm-up phase where there are no dynamic changes for the first $10000$ evaluations. This allows \POMCwp~to optimize for an extended period for the initial bound. It should be noted that the number of evaluations in the warm-up phase and our choices of $\tau$ are relatively small compared to the choice of $10eBn^2$ used in~\cite{DBLP:conf/ijcai/QianSYT17} for optimizing the static problem with a given fixed bound $B$. The results are shown in Table~\ref{tbl:IM}. For this problem, we divide the experiment, which contains a sequence of two hundred changes, into 4 parts and we perform the statistical tests for each part separately. In this way, we show how the iterative algorithms perform during the optimization process.

\subsubsection{Empirical Analysis}


\begin{table}[t]
	
	\centering
	\tiny
	\setlength{\tabcolsep}{1.3pt}
	\caption{The mean, standard deviation values and statistical tests of the partial offline error for GGA, AdGGA, POMC, and \POMCwp correspond to influence maximization problem.}
	\label{tbl:IM}
		\begin{tabular}{llllrrlrrlrrlrrl}
			& \multicolumn{1}{c}{$r$} & \multicolumn{1}{c}{$\tau$} &
			\multicolumn{1}{c}{interval} &
			\multicolumn{3}{c}{GGA (1)}     & \multicolumn{3}{c}{AdGGA (2)}                                                 & \multicolumn{3}{c}{POMC (3)}    & \multicolumn{3}{c}{\POMCwp (4)}                                                              \\
			&    &               &                & \multicolumn{1}{c}{mean} & \multicolumn{1}{c}{st} & \multicolumn{1}{c}{stat} & \multicolumn{1}{c}{mean} & \multicolumn{1}{c}{st} & \multicolumn{1}{c}{stat} & \multicolumn{1}{c}{mean} & \multicolumn{1}{c}{st} & \multicolumn{1}{c}{stat} & \multicolumn{1}{c}{mean} & \multicolumn{1}{c}{st} & \multicolumn{1}{c}{stat} \\ \midrule[.1em]
			Rout  &        1        &        100        &        1-50        &        10.05        &        1.57        &        $3^{(+)}$,$4^{(+)}$        &        11.53        &        3.76        &        $3^{(+)}$,$4^{(+)}$        &        57.93        &        20.13        &        $1^{(-)}$,$2^{(-)}$,$4^{(-)}$        &        29.25        &        9.98        &        $1^{(-)}$,$2^{(-)}$,$3^{(+)}$        \\ constraint      &        1        &        100        &        51-100        &        9.98        &        1.69        &        $3^{(+)}$,$4^{(+)}$        &        12.36        &        4.91        &        $3^{(+)}$,$4^{(+)}$        &        40.88        &        16.20        &        $1^{(-)}$,$2^{(-)}$        &        27.65        &        9.89        &        $1^{(-)}$,$2^{(-)}$        \\       &        1        &        100        &        100-151        &        9.73        &        1.48        &        $3^{(+)}$,$4^{(+)}$        &        12.25        &        4.04        &        $3^{(+)}$,$4^{(+)}$        &        36.96        &        11.91        &        $1^{(-)}$,$2^{(-)}$        &        29.11        &        11.18        &        $1^{(-)}$,$2^{(-)}$        \\       &        1        &        100        &        151-200        &        9.72        &        1.57        &        $3^{(+)}$,$4^{(+)}$        &        14.22        &        5.39        &        $3^{(+)}$,$4^{(+)}$        &        35.96        &        13.74        &        $1^{(-)}$,$2^{(-)}$        &        29.75        &        12.41        &        $1^{(-)}$,$2^{(-)}$        \\ \hline       &        1        &        1000        &        1-50        &        10.05        &        1.57        &        $3^{(+)}$,$4^{(+)}$        &        11.53        &        3.76        &        $3^{(+)}$,$4^{(+)}$        &        24.73        &        5.17        &        $1^{(-)}$,$2^{(-)}$        &        18.80        &        4.63        &        $1^{(-)}$,$2^{(-)}$        \\       &        1        &        1000        &        51-100        &        9.98        &        1.69        &        $3^{(+)}$,$4^{(+)}$        &        12.36        &        4.91        &                &        14.10        &        3.96        &        $1^{(-)}$        &        12.95        &        4.21        &        $1^{(-)}$        \\       &        1        &        1000        &        100-151        &        9.73        &        1.48        &        $2^{(+)}$,$3^{(+)}$        &        12.25        &        4.04        &        $1^{(-)}$        &        12.38        &        3.57        &        $1^{(-)}$        &        11.68        &        4.27        &                \\       &        1        &        1000        &        151-200        &        9.72        &        1.57        &        $2^{(+)}$        &        14.22        &        5.39        &        $1^{(-)}$,$3^{(-)}$,$4^{(-)}$        &        11.19        &        5.54        &        $2^{(+)}$        &        10.55        &        4.67        &        $2^{(+)}$        \\ \hline       &        1        &        5000        &        1-50        &        10.05        &        1.57        &                &        11.53        &        3.76        &        $4^{(-)}$        &        10.74        &        2.33        &                &        9.16        &        2.49        &        $2^{(+)}$        \\       &        1        &        5000        &        51-100        &        9.98        &        1.69        &        $3^{(-)}$,$4^{(-)}$        &        12.36        &        4.91        &        $3^{(-)}$,$4^{(-)}$        &        4.32        &        1.94        &        $1^{(+)}$,$2^{(+)}$        &        3.96        &        2.18        &        $1^{(+)}$,$2^{(+)}$        \\       &        1        &        5000        &        100-151        &        9.73        &        1.48        &        $3^{(-)}$,$4^{(-)}$        &        12.25        &        4.04        &        $3^{(-)}$,$4^{(-)}$        &        3.45        &        1.98        &        $1^{(+)}$,$2^{(+)}$        &        3.38        &        2.28        &        $1^{(+)}$,$2^{(+)}$        \\       &        1        &        5000        &        151-200        &        9.72        &        1.57        &        $3^{(-)}$,$4^{(-)}$        &        14.22        &        5.39        &        $3^{(-)}$,$4^{(-)}$        &        2.64        &        2.42        &        $1^{(+)}$,$2^{(+)}$        &        2.65        &        2.79        &        $1^{(+)}$,$2^{(+)}$        \\ \hline       &        1        &        10000        &        1-50        &        10.05        &        1.57        &        $3^{(-)}$,$4^{(-)}$        &        11.53        &        3.76        &        $3^{(-)}$,$4^{(-)}$        &        6.84        &        1.91        &        $1^{(+)}$,$2^{(+)}$        &        5.78        &        2.20        &        $1^{(+)}$,$2^{(+)}$        \\       &        1        &        10000        &        51-100        &        9.98        &        1.69        &        $3^{(-)}$,$4^{(-)}$        &        12.36        &        4.91        &        $3^{(-)}$,$4^{(-)}$        &        2.39        &        2.68        &        $1^{(+)}$,$2^{(+)}$        &        1.50        &        2.05        &        $1^{(+)}$,$2^{(+)}$        \\       &        1        &        10000        &        100-151        &        9.73        &        1.48        &        $3^{(-)}$,$4^{(-)}$        &        12.25        &        4.04        &        $3^{(-)}$,$4^{(-)}$        &        1.91        &        1.91        &        $1^{(+)}$,$2^{(+)}$        &        1.21        &        2.35        &        $1^{(+)}$,$2^{(+)}$        \\       &        1        &        10000        &        151-200        &        9.72        &        1.57        &        $3^{(-)}$,$4^{(-)}$        &        14.22        &        5.39        &        $3^{(-)}$,$4^{(-)}$        &        1.27        &        2.42        &        $1^{(+)}$,$2^{(+)}$        &        1.08        &        2.42        &        $1^{(+)}$,$2^{(+)}$        \\ 
			
			\midrule[.1em]
			
			Card  &        1        &        100        &        1-50        &        1.18        &        0.47        &        $3^{(+)}$,$4^{(+)}$        &        1.25        &        0.57        &        $3^{(+)}$,$4^{(+)}$        &        68.97        &        12.46        &        $1^{(-)}$,$2^{(-)}$,$4^{(-)}$        &        36.68        &        13.13        &        $1^{(-)}$,$2^{(-)}$,$3^{(+)}$        \\    constraint   &        1        &        100        &        51-100        &        1.00        &        0.55        &        $3^{(+)}$,$4^{(+)}$        &        1.15        &        0.62        &        $3^{(+)}$,$4^{(+)}$        &        51.53        &        12.07        &        $1^{(-)}$,$2^{(-)}$,$4^{(-)}$        &        34.84        &        13.08        &        $1^{(-)}$,$2^{(-)}$,$3^{(+)}$        \\       &        1        &        100        &        100-151        &        0.90        &        0.57        &        $3^{(+)}$,$4^{(+)}$        &        0.97        &        0.63        &        $3^{(+)}$,$4^{(+)}$        &        46.25        &        11.26        &        $1^{(-)}$,$2^{(-)}$        &        33.98        &        13.38        &        $1^{(-)}$,$2^{(-)}$        \\       &        1        &        100        &        151-200        &        0.75        &        0.65        &        $3^{(+)}$,$4^{(+)}$        &        0.79        &        0.72        &        $3^{(+)}$,$4^{(+)}$        &        42.91        &        12.59        &        $1^{(-)}$,$2^{(-)}$        &        32.31        &        13.04        &        $1^{(-)}$,$2^{(-)}$        \\ \hline       &        1        &        1000        &        1-50        &        1.18        &        0.47        &        $3^{(+)}$,$4^{(+)}$        &        1.25        &        0.57        &        $3^{(+)}$,$4^{(+)}$        &        31.53        &        7.18        &        $1^{(-)}$,$2^{(-)}$        &        22.93        &        4.58        &        $1^{(-)}$,$2^{(-)}$        \\       &        1        &        1000        &        51-100        &        1.00        &        0.55        &        $3^{(+)}$,$4^{(+)}$        &        1.15        &        0.62        &        $3^{(+)}$,$4^{(+)}$        &        13.97        &        5.58        &        $1^{(-)}$,$2^{(-)}$        &        12.31        &        4.36        &        $1^{(-)}$,$2^{(-)}$        \\       &        1        &        1000        &        100-151        &        0.90        &        0.57        &        $3^{(+)}$,$4^{(+)}$        &        0.97        &        0.63        &        $3^{(+)}$,$4^{(+)}$        &        10.22        &        4.59        &        $1^{(-)}$,$2^{(-)}$        &        9.55        &        3.14        &        $1^{(-)}$,$2^{(-)}$        \\       &        1        &        1000        &        151-200        &        0.75        &        0.65        &        $3^{(+)}$,$4^{(+)}$        &        0.79        &        0.72        &        $3^{(+)}$,$4^{(+)}$        &        8.28        &        5.48        &        $1^{(-)}$,$2^{(-)}$        &        7.82        &        4.29        &        $1^{(-)}$,$2^{(-)}$        \\ \hline       &        1        &        5000        &        1-50        &        1.18        &        0.47        &        $3^{(+)}$,$4^{(+)}$        &        1.25        &        0.57        &        $3^{(+)}$,$4^{(+)}$        &        10.25        &        2.65        &        $1^{(-)}$,$2^{(-)}$        &        7.88        &        2.02        &        $1^{(-)}$,$2^{(-)}$        \\       &        1        &        5000        &        51-100        &        1.00        &        0.55        &        $3^{(+)}$        &        1.15        &        0.62        &                &        1.80        &        1.02        &        $1^{(-)}$        &        1.29        &        0.95        &                \\       &        1        &        5000        &        100-151        &        0.90        &        0.57        &                &        0.97        &        0.63        &                &        1.59        &        1.29        &                &        1.29        &        0.97        &                \\       &        1        &        5000        &        151-200        &        0.75        &        0.65        &                &        0.79        &        0.72        &                &        1.09        &        1.31        &                &        1.02        &        1.23        &                \\ \hline       &        1        &        10000        &        1-50        &        1.18        &        0.47        &        $3^{(+)}$,$4^{(+)}$        &        1.25        &        0.57        &        $3^{(+)}$,$4^{(+)}$        &        5.40        &        1.41        &        $1^{(-)}$,$2^{(-)}$        &        4.17        &        1.12        &        $1^{(-)}$,$2^{(-)}$        \\       &        1        &        10000        &        51-100        &        1.00        &        0.55        &        $4^{(-)}$        &        1.15        &        0.62        &        $3^{(-)}$,$4^{(-)}$        &        0.68        &        0.63        &        $2^{(+)}$        &        0.40        &        0.58        &        $1^{(+)}$,$2^{(+)}$        \\       &        1        &        10000        &        100-151        &        0.90        &        0.57        &        $4^{(-)}$        &        0.97        &        0.63        &        $4^{(-)}$        &        0.59        &        0.56        &                &        0.31        &        0.62        &        $1^{(+)}$,$2^{(+)}$        \\       &        1        &        10000        &        151-200        &        0.75        &        0.65        &        $3^{(-)}$,$4^{(-)}$        &        0.79        &        0.72        &        $3^{(-)}$,$4^{(-)}$        &        0.19        &        0.37        &        $1^{(+)}$,$2^{(+)}$        &        0.16        &        0.60        &        $1^{(+)}$,$2^{(+)}$        
			
		\end{tabular}
\end{table}

We first investigate the influence maximization for the routing constraint that is based on the simulated networks as done for the static case in~\cite{DBLP:conf/ijcai/QianSYT17}. We consider a social network with $400$ nodes that are built using the popular Barabasi-Albert (BA) model~\cite{albert2002statistical} with edge probability $p = 0.1$. The routing network is based on the Erdos-Renyi (ER) model~\cite{erdds1959random} where each edge is presented with probability $p = 0.02$. Nodes are placed randomly in the plane and the edge costs are given by Euclidean distances. Furthermore, each chosen node has a cost of $0.1$. 

As it can be observed in Table~\ref{tbl:IM}, in the instances with $\tau=100$, both greedy algorithms are significantly better than POMCs. When the frequency of changes is that high, \POMCwp~is not even able to keep reducing the mean of errors in the third and fourth intervals. Furthermore, the impact of the warm-up phase is only noticeable in the first interval and fades when the plain POMC has more time to adapt its population.

Considering the instances with medium-frequency changes, we see a remarkable improvement in the performance of POMC and \POMCwp. Although GGA and AdGGA are still better in the first fifty changes of $\tau=1000$, the results show that POMC and \POMCwp~outperform AdGGA in the final interval and are not significantly worse than GGA anymore. In the frequency of $\tau=5000$, GGA does not beat POMCs in the first interval either. This shows the exceptional ability of POMC algorithms in providing a population so that efficiently adapts to an environment with dynamic changes in such a short time.
Note that in routing constraint, there is no significant difference in the performance of POMC and \POMCwp~when the frequency of changes is medium. However, the results depict that in some cases, such as the third interval of $\tau=1000$ and the first interval of $\tau = 5000$, \POMCwp~outperforms AdGGA while POMC is unable to do so.

When the optimization intervals increase to $\tau = 10000$, POMC and POMC$^\text{wp}$ outperform GGA and AdGGA in all the intervals. This shows that 10000 evaluations are enough for POMC to perform better than the greedy algorithms.

It should be noticed that the routing cost is affected by the structure of the graph. Let there are some nodes in the graph that can influence a considerable number of other nodes. Because of the routing constraint and their distance, the greedy algorithms might not select them together in a single solution. In other words, the greedy behavior of GGA and AdGGA might be against their performance in routing constraint. Such cases do not appear when the constraint is on the number of selected nodes in a solution. Hence, the greedy algorithms should behave much better in cardinality constraint. While our theoretical analyses show the outperformance of POMC, the interesting question is how much it takes for POMC and \POMCwp~to track the dynamic changes in practice.

To consider the cardinality constraint, we use the social news data which is collected from the social news aggregator Digg. 
The Digg dataset contains stories submitted to the platform over a period of a month, and IDs of users who voted on the popular stories.
The data consists of two tables that describe friendship links between users and the anonymized user votes on news stories~\cite{DBLP:journals/corr/abs_1202_0031}.
As in~\cite{DBLP:conf/ijcai/QianSYT17}, we use the preprocessed data with 3523 nodes and 90244 edges, and the estimated edge probabilities from the user votes based on the method in~\cite{barbieri2012topic}. 

While the effect of the warm-up phase is still visible in the first two intervals, similar to the routing constraint, POMC and \POMCwp~cannot beat greedy algorithms in cardinality constraint and high-frequency changes. The statistical advantage of greedy algorithms holds even when $\tau$ increases to $1000$. However, a significant improvement happens in $\tau = 5000$. The results show that greedy algorithms lose their advantage against \POMCwp~after the first fifty changes. But POMC needs more time to reach the other algorithms. 
The fact that evolutionary algorithms do not show a significant performance until $\tau = 5000$ shows the effectiveness of the greedy approaches in situations where the constraint is not affected by the structure of the graph, as discussed above.

Contrary to the routing constraint, POMC and \POMCwp~are outperformed by the greedy algorithms in the first interval, even when $\tau=10000$. However, the warm-up phase helps \POMCwp~to take the lead and outperform other algorithms from the second interval. On the other hand, it takes three intervals for the plain POMC to show statistically better results than the greedy algorithms. Note that the warm-up phase in only 10000 evaluations. Thus, having a prepared population, even for 10000 evaluations for such a complicated benchmark, highly improves the results of evolutionary algorithms in dynamic environment. 

\subsection{EAMC}

EAMC introduced in~\cite{bianefficient} is a recently introduced evolutionary multi-objective approach that achieves the same worst-case approximation but does not face the problem that the population size might increase exponentially during the run of the algorithms. 
The question arises whether EAMC might also provide benefits over POMC in the dynamic setting and we point out in the following major problems when using it for dynamically changing constraint bounds.
EAMC has been proven to find a $\phi$\nobreakdash-approximate solution in expected time $2en^2(n + 1)=O(n^3)$. While the definition of dominance is the only factor to control the population size in POMC, EAMC keeps only two solutions per each possible solution size. Thus, the population size does not increase to more than $2n$ in EAMC. In addition to $f(X)$, this algorithm also uses another fitness function as follow: 
\begin{align*}
g(X)= 
\begin{cases}
f(X),  &  |X|=0 \\
f(X)/(1-e^{-\alpha_f\hat c(X)/B}),  & \text{otherwise}\text{.}
\end{cases}
\end{align*}

Moreover, function $\bin(i)$ returns solutions in the population with size $i$. The polynomial bound for EAMC is the results of specific definition of $g(X)$. Note that according to this definition, one needs submodularity ratio of $f$ to use EAMC which computing its exact value might be difficult. However,~\cite{bianefficient} showed that using a lower bound $\alpha$ for $\alpha_f$ results in the approximation ratio of $(\alpha_f/2)(1-1/e^{\alpha})$.
To adopt EAMC to the dynamic environment, we update the population after each dynamic change. Since EAMC keeps and compares solutions according to their size, recording infeasible solutions can cause removing the feasible ones. Thus, the update process only removes solutions that become infeasible after the change.

\begin{algorithm}[t]
	\SetKwInOut{Input}{input}
	Update $B$\;
	Update $P$\;
	\While{stopping criterion not met}{
		Select $X$ from $P$ uniformly at random\;
		$X^\prime \leftarrow$ flip each bit of $X$ with probability $\frac{1}{n}$\;
		\If{$\hat c(X')\leq B$}{
			$i\leftarrow |X'|$\;
			\If{$\bin(i) \neq \emptyset$}{
				$P\leftarrow P\cup \{X'\}$, $U^i\leftarrow V^i\leftarrow X'$\;
			}
			\Else{
				\If{$g(X')>g(U^i)$}{
					$U^i\leftarrow X'$\;
				}
				\If{$f(X')>f(V^i)$}{
					$V^i\leftarrow X'$\;
				}
			}
			$P\leftarrow(P\setminus \bin(i))\cup\{U^i\}\cup\{V^i\}$\;
		}
		$t=t+1$\;
	}
	

	\caption{EAMC Algorithm}\label{alg:EAMC}
\end{algorithm}

\subsubsection{Behaviour of EAMC for dynamically changing constraints}

Obviously EAMC computes in time $O(n^3)$ for each newly given budget $B^*$ a $\phi$-approximation as this is the expected time to achieve this approximation when starting from scratch. We discuss some problems of EAMC when working with the dynamic setting.
The drawback of EAMC compared to POMC when working in a dynamic setting is that it does not obtain for every possible budget $b \in [0,B]$ a $\phi$-approximation.
This implies that after a dynamic change that decreases the budget, the $\phi$-approximation for the newly given budget $B^*$ may have to be recomputed. Furthermore, a change in the budget effects the function $g(X)$ used in EAMC and changes the objective values of all solutions in the population.

We illustrate the effect of reducing the budget in the following and consider a very simple instance of the knapsack problem which is submodular and we have $\alpha_f=1$. We consider an instance with only two items $e_i=(c_i, f_i)$, where $e_1=(1,1)$ and $e_2=(2,2)$ and $B=2$ initially. The search point $01$ which selects the item $e_2$ only is the optimal solution. Furthermore we have $f(01) > f(10)$ and $g(01) = 2/(1-e^{-1}) > 1/(1-e^{-1/2})=g(10)$ which implies that the optimal population computed by $EAMC$ consists of the search point $00$ and the search point $01$ as only $01$ is included in $bin(1)$. 
However, after reducing the bound to $B^*=1$, the search point $10$ is optimal and the search point $01$ is no longer feasible. Furthermore, we have $f(00)=0$. This implies that after changing the budget to $B^*=1$, the previously optimal population does not have any solution with a meaningful performance guarantee with respect to the new optimal solution $10$ that consists of the item $e_1$ only. Note that the optimal population of POMC for $B=2$ consists of the search points $00$, $01$, and $10$ as they are all feasible and do not dominate each other. After reducing the budget to $B^*=1$, POMC again contains the optimal search point, $10$ in this case.
On the positive side, EAMC reacts to such a change very quickly as it is able to produce the optimal solution $10$ for $B=1$ from the search point $00$ by flipping the first bit.

We now consider the case where the budget increases from $B$ to $B^*$ by the dynamic change. We point out that changing the budget from $B$ to $B^*$ effects the value of $g(X)$ for a given solution $X$ and therefore the final population.
We consider again two items $e'_1=(1,1)$ and $e'_2=(2.1,2)$. For $B=4$, all four possible solutions are feasible and non dominated for POMC. We have $f(01) > f(10)$ and $g(01) = 2/(1-e^{-2.1/4}) > 1/(1-e^{-1/4})=g(10)$ which implies that the optimal population of EAMC consists of the search  points $00$, $01$, and $11$.
Changing the budget to $B^*=12$, we have $g(01) = 2/(1-e^{-2.1/12}) < 1/(1-e^{-1/12})=g(10)$ which implies that the optimal population consists of $00, 01, 10, 11$ even though there is no change in the set of feasible solutions.

The previous investigations shed some light on the behavior of EAMC in the dynamic setting. However, it is an open topic for future research how quickly good approximations can be rediscovered dependent on the severity of a budget change. Especially, it would be interesting to see whether there can be small increases in the budget such that EAMC needs a sufficiently large number of steps to rediscover a $\phi$-approximation.

\subsection{NSGA-II with elitism}
The other algorithm that we use in this section for comparison is a version of NSGA-II with additional elitism with population sizes 20 and 100. NSGA-II, as explained in Algorithm~\ref{alg:nsgaii}, sorts solutions based on non-dominated fronts based on the objective values such that each solutions in front $\mathcal{F}_i$, $i>1$, is dominated by at least one solution in $\mathcal{F}_{i-1}$. It also computes the distance between solutions in the same front, called crowding distance, and uses this to achieve well-distributed solutions. Finally, front ranks and crowding distance is used to generate the next generation. Roostapour et al.~\cite{roostapour2020pareto} showed that the plain version of NSGA-II loses the best-found solution during optimization of the dynamic knapsack problem. That is the result of original approach which prioritize the distribution of solutions. The additional elitism locates the best-found solution according to the budget constraint in each generation and increases its crowding distance. Hence, it assures that this solution is not removed from the population because of the distribution factor. We use $f_{N}$ and $c_N$ as the objective functions and set the population size twenty. To prepare this algorithm for the upcoming changes, we let it to keep solutions with costs up to $B+\delta$ and we set high penalty to solutions that violate this bound. Hence, we define $f_N$ and $c_N$ as follow:
\begin{equation}
c_{N}(X)=
\begin{cases}
c(X)   & \text{if}\ c(x)\leq B+\delta \\
c(X)+(n\cdot c_{\max}+1)  \cdot h(X)   & \text{otherwise}
\end{cases}
\end{equation}
and 
\begin{equation}
f_{N}(X)=
\begin{cases}
f(X)   & \text{if}\ c(x)\leq B+\delta \\
f(X)+(n\cdot f_{\max}+1)  \cdot h(X)   & \text{otherwise,}
\end{cases}
\end{equation}
where $c_{\max}$, $f_{\max}$, and $h(X)$ are maximum possible cost, maximum possible fitness, and the amount of violation, respectively.
\begin{algorithm}[t]
	Update $B$\;
	Update objective values of solutions in population set $P_t$ and offspring set $Q_t$\;
	\While{stopping criterion not met}{
		$R_t \leftarrow P_t\cup Q_t$\tcp*[r]{combine parent and offspring population}
		$\mathcal{F} \leftarrow$ fast-non-dominated-sort($R_t$)\tcp*[r]{$\mathcal{F}=(\mathcal{F}_1,\mathcal{F}_2,\cdots)$, all non-dominated fronts of $R_t$}
		$P_{t+1}\leftarrow \emptyset$ and $i\leftarrow 1$\;
		\While{$|P_{t+1}|+|\mathcal{F}_i|\leq N$}{
			crowding-distance-assignment($\mathcal{F}_i$)\;
			$P_{t+1} \leftarrow P_{t+1} \cup \mathcal{F}_i$\;
			$i\leftarrow i+1$\;
		}
		Sort $\mathcal{F}_i$ based on the crowding distance in descending order\;
		$P_{t+1} \leftarrow P_{t+1}\cup \mathcal{F}_i[1:(N-|P_{t+1}|)]$\;
		\label{lin:elit-nsga}$Q_{t+1}\leftarrow$ make-new-pop($P_{t+1})$\;
		$t\leftarrow t+1$\;
	}
	\caption{NSGA-II}\label{alg:nsgaii}
\end{algorithm}
\subsection{The Maximum Coverage Problem}\label{sec:MaxCov}
We use the maximum coverage problem to compare the performance of AdGGA, GGA, \POMCwp, EAMC, and NSGA-II, which have been described in previous sections. In this problem, a set of elements $U$ and a collection $V=\{S_1,S_2,\cdots,S_n\}$ of subsets of $U$ are given. Considering a monotone cost function $c$, the goal is to select subsets from $V$ such that their union cover the maximum number of elements in $U$, while the cost of this selection does not exceed the budget constraint $B$. To be more precise, we are looking for $\argmax_{X\subseteq V}f(X)=|\bigcup_{S_i\in X}S_i|$ such that $c(X)\leq B$. We use directed graphs as the benchmarks for this problem. Each node $p$ represents a subset $S_p\in V$ that contains node $p$ and all of its adjacent nodes. We use two types of cost functions as defined by~\cite{bianefficient}. The ``outdegree'' cost of node $p$ is calculated as $o(p) = 1-\max\{d(p)-q,0\}$, where $d(p)$ is the out degree of $p$ and $q$ is a constant set to six. For ``random'' cost function we assigned a random positive cost value in $(0,1]$ to each node. Cost of a selected set of nodes $X$ is $c(X)=\sum_{p\in X}o(p)$. Our benchmarks are two graphs originally generated for maximum independent set problem~\cite{DBLP:journals/ai/XuBHL07}. frb35-17-mis is a graph with 450 nodes and 17,827 edges, and the graph of frb30-15-mis has 595 nodes and 27,856 edges.

To apply the dynamic changes, similar to Section~\ref{sec:InfMax}, we generated thirty files for each constraint that includes two hundred random numbers. However, the magnitudes of changes, budget intervals, and initial budget are chosen according to each cost function. The budget for ``outdegree'' is initially set to 500 that is bounded by the $B_{\min} =250$ and $B_{\max} = 750$ during the optimization process, and $\delta$ is a random integer within $[-20,20]$. In instances with ``random'' cost, the budget changes within the range $[0,3]$, initially set to $1$ and the magnitude of changes is $\delta\in\{-0.1,0.1\}$. The results are presented in Tables~\ref{tbl:MC-GGA} and~\ref{tbl:MC-EAMC}.

\subsubsection{Empirical Analysis}
In this section, we compare the performance of three evolutionary and two greedy algorithms in dynamic environments with a variety of frequencies. Moreover, we have closely observe the changes of \POMCwp population size during a single dynamic benchmark. In Section~\ref{sec:InfMax}, we divided the sequence of dynamic changes into four intervals to study how our evolutionary algorithms improve during the optimization process. In this experiment, we consider all 200 changes as a whole, and we calculate the partial offline error for all the changes. Since smaller benchmarks are considered here, we can study how long does it take for the algorithms to compensate the errors generated at the beginning of the dynamic changes and outperform the others after 200 changes.

Note that the random cost is similar to the cardinality constraint in Section~\ref{sec:InfMax} since it considers each node despite the structure of the graph.  In other words, the increase in fitness value occurred by adding a specific node is independent of the rise in the cost. Adding a node that covers lots of other nodes, however, is always  expensive when we consider the outdegree cost. Thus, we expect better results from the greedy approaches when the random constraint is considered.


\begin{table}
	\centering
	\tiny
	\setlength{\tabcolsep}{1pt}
	\caption{The mean, standard deviation values and statistical tests of the partial offline error for GGA, AdGGA, \POMCwp, EAMC, NSGA-II20, and NSGA-II100 with elitism correspond to maximum coverage problem, based on the number of evaluations.}
	\label{tbl:MC-GGA}
		\begin{tabular}{llllrrlrrlrrl}
			& \multicolumn{1}{c}{$n$} & \multicolumn{1}{c}{$r$} & \multicolumn{1}{c}{$\tau$} & \multicolumn{3}{c}{GGA (1)}     & \multicolumn{3}{c}{AdGGA (2)}    & \multicolumn{3}{c}{\POMCwp (3)}\\
			&                            &                       &                            & \multicolumn{1}{c}{mean} & \multicolumn{1}{c}{st} & \multicolumn{1}{c}{stat} & \multicolumn{1}{c}{mean} & \multicolumn{1}{c}{st} & \multicolumn{1}{c}{stat} & \multicolumn{1}{c}{mean} & \multicolumn{1}{c}{st} & \multicolumn{1}{c}{stat} \\ \hline
			frb30-15   & 450                        & 0.1                  & 100                        &        2.12         &        0.55         &        $3^{(+)}$,$4^{(+)}$,$5^{(+)}$,$6^{(+)}$          &          8.92         &        1.30         &        $3^{(+)}$,$4^{(+)}$,$6^{(+)}$          &          33.87         &        13.35         &        $1^{(-)}$,$2^{(-)}$           \\
			random    & 450                        & 0.1                  & 1000                       &        2.12         &        0.55         &        $2^{(+)}$,$3^{(+)}$,$4^{(+)}$,$6^{(+)}$          &          8.92         &        1.30         &        $1^{(-)}$,$4^{(+)}$,$5^{(-)}$         &          10.60         &        3.94         &        $1^{(-)}$,$4^{(+)}$,$5^{(-)}$                \\
			& 450                        & 0.1                  & 5000                       &        2.12         &        0.55         &        $2^{(+)}$,$3^{(+)}$,$4^{(+)}$          &          8.92         &        1.30         &        $1^{(-)}$,$3^{(-)}$,$5^{(-)}$,$6^{(-)}$          &          3.60         &        1.02         &        $1^{(-)}$,$2^{(+)}$,$4^{(+)}$,$5^{(-)}$          \\
			& 450                        & 0.1                  & 15000                      &        2.12         &        0.55         &        $2^{(+)}$,$4^{(+)}$,$5^{(-)}$,$6^{(-)}$          &          8.92         &        1.30         &        $1^{(-)}$,$3^{(-)}$,$5^{(-)}$,$6^{(-)}$          &          1.55         &        0.31         &        $2^{(+)}$,$4^{(+)}$,$6^{(-)}$                 \\
			& 450                        & 0.1                  & 45000                &        2.12         &        0.55         &        $2^{(+)}$,$3^{(-)}$,$5^{(-)}$,$6^{(-)}$          &          8.92         &        1.30         &        $1^{(-)}$,$3^{(-)}$,$5^{(-)}$,$6^{(-)}$          &          0.67         &        0.14         &        $1^{(+)}$,$2^{(+)}$,$4^{(+)}$          \\   \hline
			frb30-15   & 450                        & 20                  & 100                        &        16.66         &        1.54         &        $2^{(+)}$,$3^{(+)}$,$4^{(+)}$,$5^{(+)}$,$6^{(+)}$          &          23.83         &        5.32         &        $1^{(-)}$,$4^{(+)}$          &          28.75         &        3.21         &        $1^{(-)}$,$4^{(+)}$,$6^{(-)}$      \\
			outdegree    & 450                        & 20                &  1000  &        16.66         &        1.54         &        $2^{(+)}$,$4^{(+)}$,$6^{(-)}$          &          23.83         &        5.32         &        $1^{(-)}$,$3^{(-)}$,$6^{(-)}$          &          16.36         &        2.46         &        $2^{(+)}$,$4^{(+)}$,$6^{(-)}$ \\
			& 450                        & 20                  & 5000                       &        16.66         &        1.54         &        $3^{(-)}$,$4^{(+)}$,$6^{(-)}$          &          23.83         &        5.32         &        $3^{(-)}$,$5^{(-)}$,$6^{(-)}$          &          9.01         &        2.11         &        $1^{(+)}$,$2^{(+)}$,$4^{(+)}$ \\
			& 450                        & 20                  & 15000                      &        16.66         &        1.54         &        $3^{(-)}$,$5^{(-)}$,$6^{(-)}$          &          23.83         &        5.32         &        $3^{(-)}$,$5^{(-)}$,$6^{(-)}$          &          6.14         &        1.75         &        $1^{(+)}$,$2^{(+)}$,$4^{(+)}$ \\
			& 450                        & 20                  & 45000          &        16.66         &        1.54         &        $2^{(+)}$,$3^{(-)}$,$5^{(-)}$,$6^{(-)}$          &          23.83         &        5.32         &        $1^{(-)}$,$3^{(-)}$,$5^{(-)}$,$6^{(-)}$          &          4.25         &        1.87         &        $1^{(+)}$,$2^{(+)}$,$4^{(+)}$  \\ \hline
			frb35-17   & 595                        & 0.1                  & 100                &        2.28         &        0.91         &        $3^{(+)}$,$4^{(+)}$,$5^{(+)}$,$6^{(+)}$          &          6.25         &        2.61         &        $3^{(+)}$,$4^{(+)}$,$6^{(+)}$          &          54.35         &        21.71         &        $1^{(-)}$,$2^{(-)}$   \\
			random    & 595                        & 0.1                  & 1000              &        2.28         &        0.91         &        $3^{(+)}$,$4^{(+)}$,$5^{(+)}$,$6^{(+)}$          &          6.25         &        2.61         &        $3^{(+)}$,$4^{(+)}$,$6^{(+)}$          &          18.58         &        8.93         &        $1^{(-)}$,$2^{(-)}$,$4^{(+)}$            \\
			& 595                        & 0.1                  & 5000              &        2.28         &        0.91         &        $2^{(+)}$,$3^{(+)}$,$4^{(+)}$,$6^{(+)}$          &          6.25         &        2.61         &        $1^{(-)}$,$4^{(+)}$,$5^{(-)}$          &          6.32         &        2.67         &        $1^{(-)}$,$4^{(+)}$,$5^{(-)}$          \\
			&595                        & 0.1                  & 15000              &        2.28         &        0.91         &        $2^{(+)}$,$4^{(+)}$,$5^{(-)}$          &          6.25         &        2.61         &        $1^{(-)}$,$3^{(-)}$,$5^{(-)}$,$6^{(-)}$          &          2.36         &        0.82         &        $2^{(+)}$,$4^{(+)}$,$5^{(-)}$           \\
			& 595                        & 0.1                  & 45000              &        2.28         &        0.91         &        $3^{(-)}$,$4^{(+)}$,$5^{(-)}$,$6^{(-)}$          &          6.25         &        2.61         &        $3^{(-)}$,$5^{(-)}$,$6^{(-)}$          &          0.73         &        0.25         &        $1^{(+)}$,$2^{(+)}$,$4^{(+)}$           \\ \hline
			frb35-17   & 595                        & 20                  & 100               &        16.00         &        1.08         &        $2^{(+)}$,$3^{(+)}$,$4^{(+)}$,$5^{(+)}$,$6^{(+)}$          &          35.24         &        10.69         &        $1^{(-)}$,$4^{(+)}$          &          38.48         &        3.71         &        $1^{(-)}$,$4^{(+)}$,$6^{(-)}$           \\
			outdegree    & 595                        & 20                  & 1000          &        16.00         &        1.08         &        $2^{(+)}$,$4^{(+)}$,$5^{(+)}$          &          35.24         &        10.69         &        $1^{(-)}$,$3^{(-)}$,$6^{(-)}$          &          21.25         &        2.73         &        $2^{(+)}$,$4^{(+)}$,$6^{(-)}$          \\
			& 595                        & 20                  & 5000             &        16.00         &        1.08         &        $2^{(+)}$,$3^{(-)}$,$4^{(+)}$,$6^{(-)}$          &          35.24         &        10.69         &        $1^{(-)}$,$3^{(-)}$,$5^{(-)}$,$6^{(-)}$          &          9.77         &        2.42         &        $1^{(+)}$,$2^{(+)}$,$4^{(+)}$,$5^{(+)}$          \\
			& 595                        & 20                  & 15000                &        16.00         &        1.08         &        $2^{(+)}$,$3^{(-)}$,$4^{(+)}$,$6^{(-)}$          &          35.24         &        10.69         &        $1^{(-)}$,$3^{(-)}$,$5^{(-)}$,$6^{(-)}$          &          4.93         &        2.40         &        $1^{(+)}$,$2^{(+)}$,$4^{(+)}$          \\
			& 595                        & 20                  & 45000                &        16.00         &        1.08         &        $2^{(+)}$,$3^{(-)}$,$6^{(-)}$          &          35.24         &        10.69         &        $1^{(-)}$,$3^{(-)}$,$5^{(-)}$,$6^{(-)}$          &          2.12         &        1.86         &        $1^{(+)}$,$2^{(+)}$,$4^{(+)}$,$5^{(+)}$          \\                         
		\end{tabular}
\end{table}

The results of 6 different approaches in terms of mean, standard deviation and statistical test are shown in Tables~\ref{tbl:MC-GGA} and~\ref{tbl:MC-EAMC}.
Generally observing the results in Table 2, AdGGA performed worse than GGA in most of the cases. It should be noticed that its inability to trace the best solution is already proven in Section~\ref{sec:AdGGA}.  However, the results confirm that in comparison with the evolutionary approaches, AdGGA's greedy behavior is still beneficial in the environment that the frequency of changes is high. On the other hand, in Table 3, we have EAMC that has the worst performance among the other iterative algorithms in most of the cases. It does not statistically outperform any of the algorithms in our experiments. Considering the random cost, it achieves a lower mean of partial error only in the smaller benchmark and $\tau=45000$. While 45000 is still significantly less than the proven theoretical expected time $O(n^3)$, EAMC cannot compete with the other algorithms. The situation becomes a bit better with the outdegree cost. Although it is still unable to outperform the others, the mean of its partial offline error becomes lower than AdGGA's in $\tau=15000$. The reason for such a poor performance is how EAMC compares solutions to keep them in the population. Despite the fact that its approach guarantees the polynomial size of the population,  keeping a solution only because no better solution with the same size has been found increases the chance of storing a bad solution. Hence, it needs more time to prepare its population for the next dynamic change. In the following, we compare the performance of the rest of the algorithms in detail.

\begin{table}
	\centering
	\tiny
	\setlength{\tabcolsep}{1pt}
	\caption{The mean, standard deviation values and statistical tests of the partial offline error for GGA, AdGGA, \POMCwp, EAMC, NSGA-II20, and NSGA-II100 with elitism correspond to maximum coverage problem, based on the number of evaluations.}
	\label{tbl:MC-EAMC}
		\begin{tabular}{llllrrlrrlrrl}
			& \multicolumn{1}{c}{$n$} & \multicolumn{1}{c}{$r$} & \multicolumn{1}{c}{$\tau$} & \multicolumn{3}{c}{EAMC (4)}   & \multicolumn{3}{c}{NSGA-II20(5)} & \multicolumn{3}{c}{NSGA-II100(6)}\\
			&                            &                       &                            & \multicolumn{1}{c}{mean} & \multicolumn{1}{c}{st} & \multicolumn{1}{c}{stat} & \multicolumn{1}{c}{mean} & \multicolumn{1}{c}{st} & \multicolumn{1}{c}{stat} & \multicolumn{1}{c}{mean} & \multicolumn{1}{c}{st} & \multicolumn{1}{c}{stat} \\ \hline
			frb30-15   & 450                        & 0.1                  & 100                        &          59.57         &        24.18         &        $1^{(-)}$,$2^{(-)}$,$5^{(-)}$         &          19.78         &        9.96         &        $1^{(-)}$,$4^{(+)}$,$6^{(+)}$          &          36.86         &        8.53         &        $1^{(-)}$,$2^{(-)}$,$5^{(-)}$      \\
			random    & 450                        & 0.1                  & 1000                       &                28.16         &        9.14         &        $1^{(-)}$,$2^{(-)}$,$3^{(-)}$,$5^{(-)}$,$6^{(-)}$          &          5.27         &        1.86         &        $2^{(+)}$,$3^{(+)}$,$4^{(+)}$,$6^{(+)}$          &          8.40         &        1.91         &        $1^{(-)}$,$4^{(+)}$,$5^{(-)}$       \\
			& 450                        & 0.1                  & 5000                       &          15.20         &        4.86         &        $1^{(-)}$,$3^{(-)}$,$5^{(-)}$,$6^{(-)}$          &          2.15         &        0.75         &        $2^{(+)}$,$3^{(+)}$,$4^{(+)}$          &          2.37         &        0.41         &        $2^{(+)}$,$4^{(+)}$       \\
			& 450                        & 0.1                  & 15000                      &          10.20         &        3.37         &        $1^{(-)}$,$3^{(-)}$,$5^{(-)}$,$6^{(-)}$          &          1.13         &        0.33         &        $1^{(+)}$,$2^{(+)}$,$4^{(+)}$          &          1.00         &        0.27         &        $1^{(+)}$,$2^{(+)}$,$3^{(+)}$,$4^{(+)}$       \\
			& 450                        & 0.1                  & 45000                &          6.76         &        2.29         &        $3^{(-)}$,$5^{(-)}$,$6^{(-)}$          &          0.59         &        0.23         &        $1^{(+)}$,$2^{(+)}$,$4^{(+)}$          &          0.53         &        0.18         &        $1^{(+)}$,$2^{(+)}$,$4^{(+)}$       \\   \hline
			frb30-15   & 450                        & 20                  & 100         &          43.18         &        6.57         &        $1^{(-)}$,$2^{(-)}$,$3^{(-)}$,$5^{(-)}$,$6^{(-)}$          &          28.18         &        3.62         &        $1^{(-)}$,$4^{(+)}$,$6^{(-)}$          &          23.09         &        3.71         &        $1^{(-)}$,$3^{(+)}$,$4^{(+)}$,$5^{(+)}$       \\
			outdegree    & 450                        & 20                &  1000    &          32.54         &        6.88         &        $1^{(-)}$,$3^{(-)}$,$5^{(-)}$,$6^{(-)}$          &          19.48         &        2.32         &        $4^{(+)}$,$6^{(-)}$          &          11.55         &        2.42         &        $1^{(+)}$,$2^{(+)}$,$3^{(+)}$,$4^{(+)}$,$5^{(+)}$       \\
			& 450                        & 20                  & 5000             &          27.69         &        6.31         &        $1^{(-)}$,$3^{(-)}$,$5^{(-)}$,$6^{(-)}$          &          12.58         &        2.56         &        $2^{(+)}$,$4^{(+)}$,$6^{(-)}$          &          8.20         &        1.57         &        $1^{(+)}$,$2^{(+)}$,$4^{(+)}$,$5^{(+)}$       \\
			& 450                        & 20                  & 15000                      &          22.75         &        6.16         &        $3^{(-)}$,$5^{(-)}$,$6^{(-)}$          &          9.95         &        2.14         &        $1^{(+)}$,$2^{(+)}$,$4^{(+)}$,$6^{(-)}$          &          5.82         &        2.11         &        $1^{(+)}$,$2^{(+)}$,$4^{(+)}$,$5^{(+)}$       \\
			& 450                        & 20                  & 45000          &          21.71         &        3.23         &        $3^{(-)}$,$5^{(-)}$,$6^{(-)}$          &          7.89         &        2.03         &        $1^{(+)}$,$2^{(+)}$,$4^{(+)}$          &          5.02         &        2.53         &        $1^{(+)}$,$2^{(+)}$,$4^{(+)}$       \\ \hline
			frb35-17   & 595                        & 0.1                  & 100            &          100.95         &        45.88         &        $1^{(-)}$,$2^{(-)}$,$5^{(-)}$          &          31.27         &        13.57         &        $1^{(-)}$,$4^{(+)}$,$6^{(+)}$          &          73.69         &        11.27         &        $1^{(-)}$,$2^{(-)}$,$5^{(-)}$       \\
			random    & 595                        & 0.1                  & 1000           &          43.73         &        17.72         &        $1^{(-)}$,$2^{(-)}$,$3^{(-)}$,$5^{(-)}$          &          9.49         &        4.08         &        $1^{(-)}$,$4^{(+)}$,$6^{(+)}$          &          20.82         &        7.19         &        $1^{(-)}$,$2^{(-)}$,$5^{(-)}$       \\
			& 595                        & 0.1                  & 5000            &          22.13         &        7.55         &        $1^{(-)}$,$2^{(-)}$,$3^{(-)}$,$5^{(-)}$,$6^{(-)}$          &          3.38         &        1.04         &        $2^{(+)}$,$3^{(+)}$,$4^{(+)}$,$6^{(+)}$          &          5.62         &        1.80         &        $1^{(-)}$,$4^{(+)}$,$5^{(-)}$       \\
			&595                        & 0.1                  & 15000           &          13.73         &        4.21         &        $1^{(-)}$,$3^{(-)}$,$5^{(-)}$,$6^{(-)}$          &          1.40         &        0.43         &        $1^{(+)}$,$2^{(+)}$,$3^{(+)}$,$4^{(+)}$          &          1.71         &        0.31         &        $2^{(+)}$,$4^{(+)}$       \\
			& 595                        & 0.1                  & 45000               &          8.53         &        2.23         &        $1^{(-)}$,$3^{(-)}$,$5^{(-)}$,$6^{(-)}$          &          0.62         &        0.25         &        $1^{(+)}$,$2^{(+)}$,$4^{(+)}$          &          0.66         &        0.17         &        $1^{(+)}$,$2^{(+)}$,$4^{(+)}$       \\ \hline
			frb35-17   & 595                        & 20                  & 100        &          56.11         &        6.40         &        $1^{(-)}$,$2^{(-)}$,$3^{(-)}$,$5^{(-)}$,$6^{(-)}$          &          37.82         &        4.50         &        $1^{(-)}$,$4^{(+)}$,$6^{(-)}$          &          26.64         &        5.88         &        $1^{(-)}$,$3^{(+)}$,$4^{(+)}$,$5^{(+)}$       \\
			outdegree    & 595                        & 20                  & 1000   &          46.23         &        7.88         &        $1^{(-)}$,$3^{(-)}$,$5^{(-)}$,$6^{(-)}$          &          25.41         &        3.45         &        $1^{(-)}$,$4^{(+)}$,$6^{(-)}$          &          11.09         &        3.77         &        $2^{(+)}$,$3^{(+)}$,$4^{(+)}$,$5^{(+)}$       \\
			& 595                        & 20                  & 5000          &          36.30         &        8.20         &        $1^{(-)}$,$3^{(-)}$,$5^{(-)}$,$6^{(-)}$          &          16.07         &        3.23         &        $2^{(+)}$,$3^{(-)}$,$4^{(+)}$,$6^{(-)}$          &          6.36         &        2.36         &        $1^{(+)}$,$2^{(+)}$,$4^{(+)}$,$5^{(+)}$       \\
			& 595                        & 20                  & 15000              &          31.93         &        6.15         &        $1^{(-)}$,$3^{(-)}$,$5^{(-)}$,$6^{(-)}$          &          11.96         &        2.86         &        $2^{(+)}$,$4^{(+)}$,$6^{(-)}$          &          3.15         &        1.47         &        $1^{(+)}$,$2^{(+)}$,$4^{(+)}$,$5^{(+)}$       \\
			& 595                        & 20                  & 45000              &          27.00         &        4.74         &        $3^{(-)}$,$5^{(-)}$,$6^{(-)}$          &          8.19         &        2.51         &        $2^{(+)}$,$3^{(-)}$,$4^{(+)}$,$6^{(-)}$          &          1.54         &        1.42         &        $1^{(+)}$,$2^{(+)}$,$4^{(+)}$,$5^{(+)}$       \\                         
		\end{tabular}
\end{table}

Consider the results for smaller benchmark, frb30-15. In instances with $\tau=100$ and random cost, GGA and AdGGA are significantly better than the other algorithms. Moreover,  NSGA-II20 achieves better results than \POMCwp. The benefit of the greedy approach in the random cost, do not let NSGA-II20 and NSGA-II100 dominate GGA before doing 15000 evaluations. For \POMCwp, it is more challenging to get better than greedy algorithms. It requires 15000 evaluations before each change to perform as good as NSGA-II20 and NSGA-II100, and 45000 evaluations to outperform GGA. However, the situation is different when the cost function is calculated based on the outdegree. $\tau=5000$ and $\tau=15000$ are enough for \POMCwp and NSGA-II20, respectively, to outperform the greedy algorithms. In contrast, NSGA-II100 outperforms the greedy algorithms in $\tau=1000$.

For the more complicated benchmark, frb35-17, at most the same results hold. However, it gets harder for our evolutionary algorithms to beat greedy algorithms. 
For the random cost, AdGGA outperforms NSGA-II100, in addition to \POMCwp and EAMC, when $\tau=100$ and GGA remains unbeatable until $\tau= 45000$. This clearly shows that larger population slows down the NSGA-II when the frequency changes are high.

GGA also gains better results with outdegree cost. It outperforms other algorithms in $\tau= 100$, and loses only to POMC and NSGA-II100 after $\tau= 5000$. These results are reasonable since the size of solutions increases to 595 bits in the frb35-17 benchmark, which increases the expected time for each bit to be flipped by an evolutionary algorithm. 
Consequently, our algorithms need more time to find better solutions. However, the results show that NSGA-II20 seems to have more difficulties in dealing with the bigger size than \POMCwp. As expected NSGA-II100 is able to take advantage of the bigger population size and outperforms greedy algorithm, EMAC and NSGA-II20 from $\tau = 5000$ outwards.

The other fact extracted from the results is that NSGA-II20 has a better performance than EAMC for all $\tau$ values, and NSGA-II100 until $\tau = 5000$ when we have random cost.
On the other hand, the results of \POMCwp~are better compared to EAMC for all $\tau$ values, and to NSGA-II20 only for $\tau = 5000, 45000$ at outdegree cost.

This pattern can be observed in both benchmarks. That is the result of the population size and how each algorithm handles the replacement of a new solution. While \POMCwp~can store a valuable solution for each possible cost, the NSGA-II20 and NSGA-II100 concentrates on producing a well-distributed population. Consider the cases that there are a lot of non-dominated solutions with cost values close to the constraint. \POMCwp~can keep them all, which is an advantage if the next dynamic change does not affect the constraint significantly. However, because of the distribution factor, NSGA-II20 and NSGA-II100 minimize the number of solutions that are close to each other. These situations rarely happen with random constraint, since, as discussed previously, it is similar to the cardinality constraint and the possible values for the cost of the solutions are more evenly distributed. 

NSGA-II100 outperforms the greedy algorithm, EMAC and NSGA-II20 from $\tau = 1000$ at outdegree cost in both benchmarks. The results confirm that the population size is an important setting parameter. The large population size increases the diversity of a population i.e., higher number of different solutions. In contrast, NSGA-II20 leads to a crowded population, where the number of similar solutions is higher. This normally causes premature convergence i.e. NSGA-II20 is not able to generate offspring that are superior to their parents.

\begin{figure}[t]
	\centering
	\includegraphics[width=1\textwidth]{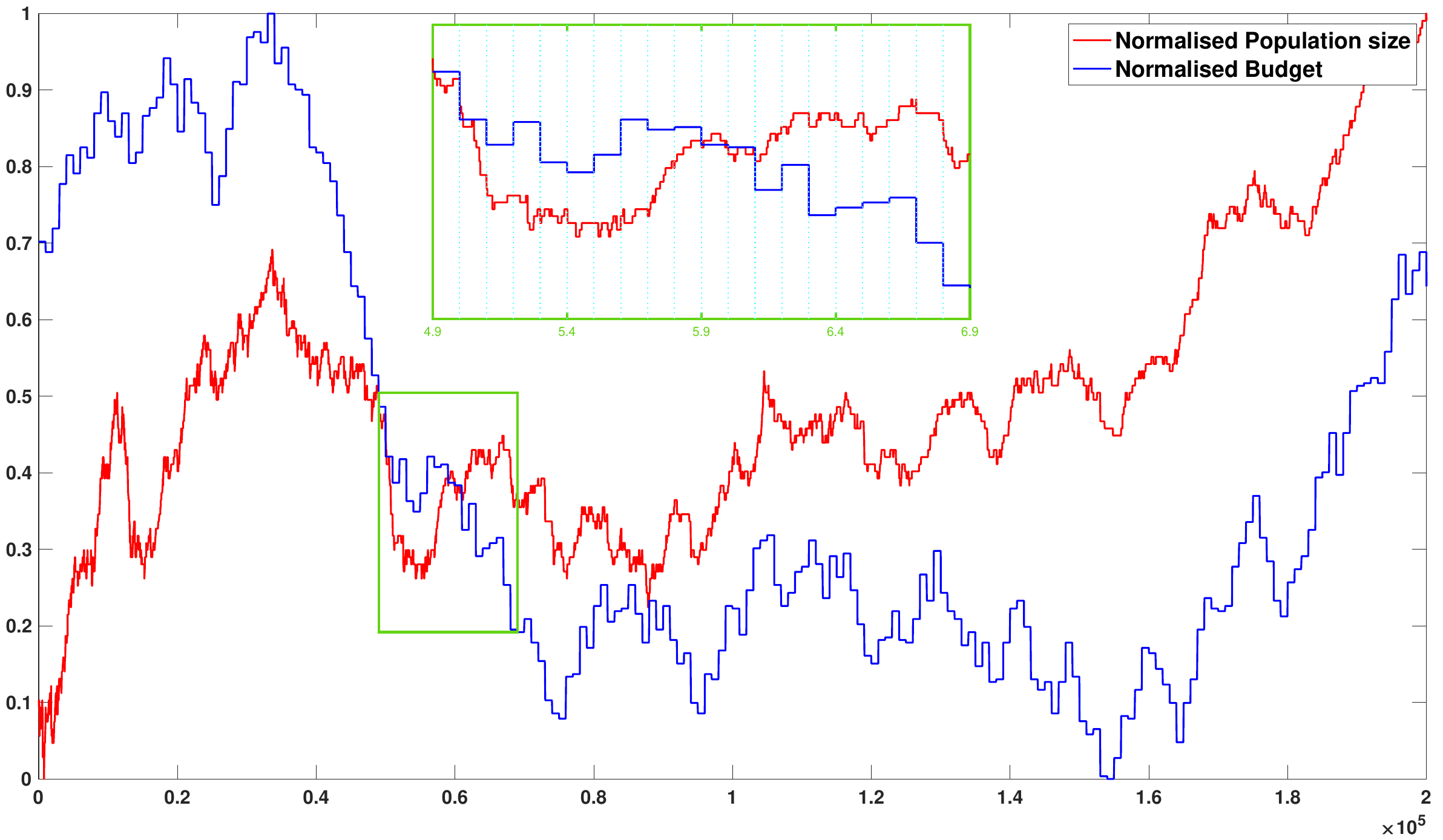}
	\caption{Normalized budget and population size for \POMCwp algorithm performing on frb30-15 instance with $\tau = 1000$ under outdegree cost. $Y$ axis shows the normalized value of budget and population size, and $X$ axis demonstrates the number of generations.}
	\label{fig:Pomc-pop}
\end{figure}

Next, we take a closer look on how the population size of \POMCwp~changes during one dynamic benchmark.
Figure~\ref{fig:Pomc-pop} presents the normalized values of the dynamic budget and population size of \POMCwp~for each generation when the algorithm is performing on the frb30-15 instance under outdegree cost function and a dynamic change happens every 200 generation. The maximum and minimum of the population size before normalization is 234 and 127, respectively. Moreover, the budget cost has changed within [295,587]. The figure clearly shows the impact of dynamic changes on the population size of \POMCwp. 

As expected, decreasing the budget value results in a drop in the population size. However, larger alteration in a budget does not necessarily cause loosing more solutions. As an example, consider the magnified part of the figure. Each vertical line in this figure demonstrates where the dynamic change happens. Although the magnitude of dynamic changes in generations 53000 and 54000 are significantly different, the number of solutions that violate the budget constraint in result of those changes are almost the same. 

It should be noted that dynamic changes are not the only reasons for the alterations in \POMCwp population size during optimization process. As mentioned in Section \ref{Sec:Algs}, \POMCwp controls the quality of the stored solution as well as its population size by storing non-dominated solutions. In other words, when a new non-dominated solution is found the algorithm removes the existing solutions that are dominated by the new one from the population. Having a closer look at the magnified part of Figure \ref{fig:Pomc-pop}, it can be observed that the population size could be decreased/increased significantly even while the budget has not changed. Examples could be generation 50000-51000 for decreasing and 57000-58000 for increasing the population size.

Figure~\ref{fig:Pomc-pop} also shows how \POMCwp algorithm improves the quality of population during the optimization period. During the sequential decrease that is started from the generation 36000 until 77000 (i.e. 36th change until 77th) it can be seen that \POMCwp has maintained the population such that the impact of dynamic changes is controlled. Moreover, although the budget is gradually decreasing before the 154th change (i.e. generation 154000), the population size is growing in comparison with its value in the 55th change (i.e. generation 55000). This shows \POMCwp has found solutions with higher quality that can produce more robust population. 

\section{Conclusions}
Many real-world problems can be modeled as submodular functions and have problems with dynamically changing constraints. We have contributed to the area of submodular optimization with dynamic constraints. Key to the investigations have been the adaptability of algorithms when constraints change. We have shown that an adaptive version of the generalized greedy algorithm frequently considered in submodular constrained optimization is not able to maintain a $\phi$\nobreakdash-approximation. Furthermore, we have pointed out that the population-based POMC algorithm is able to cater for and recompute $\phi$\nobreakdash-approximations for related constraints bounds in an efficient way. We challenged POMC and greedy approaches against EAMC as a recently introduced algorithm with polynomial expected running time, and NSGA-II with different population sizes as advanced multi-objective algorithms in practice. Our experimental results confirm the advantage of POMC and NSGA-II over the considered greedy approaches on important real-world problems. Furthermore, the experiments show that evolutionary algorithms are able to significantly improve their performance over time when dealing with dynamic changes.

\section{Acknowledgements}
We thank Chao Qian for providing his POMC and EAMC implementations and test data to carry out our experimental investigations.
This research has been supported by the Australian Research Council (ARC) through grant DP160102401 and the German Science Foundation (DFG) through grant FR2988 (TOSU). Frank Neumann has been support by the Alexander von Humboldt Foundation through a Humboldt Fellowship for Experienced Researchers.

	\bibliographystyle{abbrv}
\bibliography{ref}

\begin{thebibliography}{10}

\bibitem{albert2002statistical}
R.~Albert and A.-L. Barab{\'a}si.
\newblock Statistical mechanics of complex networks.
\newblock {\em Reviews of modern physics}, 74(1):47, 2002.

\bibitem{barbieri2012topic}
N.~Barbieri, F.~Bonchi, and G.~Manco.
\newblock Topic-aware social influence propagation models.
\newblock In {\em Proceedings of the IEEE Conference on Data Mining}, pages
  81--90. IEEE Computer Society, 2012.

\bibitem{bianefficient}
C.~Bian, C.~Feng, C.~Qian, and Y.~Yu.
\newblock An efficient evolutionary algorithm for subset selection with general
  cost constraints.
\newblock In {\em Proceedings of the Thirty-Fourth {AAAI} Conference on
  Artificial Intelligence, {AAAI} 2020}, pages 3267--3274. AAAI Press, 2020.

\bibitem{DBLP:conf/gecco/Bossek0PS19}
J.~Bossek, F.~Neumann, P.~Peng, and D.~Sudholt.
\newblock Runtime analysis of randomized search heuristics for dynamic graph
  coloring.
\newblock In {\em Proceedings of the Genetic and Evolutionary Computation
  Conference, {GECCO}}, pages 1443--1451. {ACM}, (2019).

\bibitem{DBLP:journals/dam/ConfortiC84}
M.~Conforti and G.~Cornu{\'{e}}jols.
\newblock Submodular set functions, matroids and the greedy algorithm: Tight
  worst-case bounds and some generalizations of the rado-edmonds theorem.
\newblock {\em Discrete Applied Mathematics}, 7(3):251--274, 1984.

\bibitem{Corder09}
G.~W. Corder and D.~I. Foreman.
\newblock {\em {Nonparametric Statistics for Non-Statisticians: A Step-by-Step
  Approach}}.
\newblock Wiley, 2009.

\bibitem{DBLP:journals/tec/DebAPM02}
K.~Deb, S.~Agrawal, A.~Pratap, and T.~Meyarivan.
\newblock A fast and elitist multiobjective genetic algorithm: {NSGA-II}.
\newblock {\em {IEEE} Transactions on Evolutionary Computation}, 6(2):182--197,
  2002.

\bibitem{DBLP:journals/ec/DoerrHK11}
B.~Doerr, E.~Happ, and C.~Klein.
\newblock Tight analysis of the {(1+1)-EA} for the single source shortest path
  problem.
\newblock {\em Evolutionary Computation}, 19(4):673--691, 2011.

\bibitem{erdds1959random}
P.~Erd{\H{o}}s and A.~R{\'e}nyi.
\newblock On random graphs {I}.
\newblock {\em Publicationes Mathematicae Debrecen}, 6:290--297, 1959.

\bibitem{DBLP:journals/ec/FriedrichHHNW10}
T.~Friedrich, J.~He, N.~Hebbinghaus, F.~Neumann, and C.~Witt.
\newblock Approximating covering problems by randomized search heuristics using
  multi-objective models.
\newblock {\em Evolutionary Computation}, 18(4):617--633, 2010.

\bibitem{DBLP:journals/ec/FriedrichN15}
T.~Friedrich and F.~Neumann.
\newblock Maximizing submodular functions under matroid constraints by
  evolutionary algorithms.
\newblock {\em Evolutionary Computation}, 23(4):543--558, 2015.

\bibitem{DBLP:journals/corr/abs_1202_0031}
T.~Hogg and K.~Lerman.
\newblock Social dynamics of {D}igg.
\newblock {\em EPJ Data Science}, 1(1):5, 2012.

\bibitem{DBLP:conf/nips/IyerB13}
R.~K. Iyer and J.~A. Bilmes.
\newblock Submodular optimization with submodular cover and submodular knapsack
  constraints.
\newblock In {\em Advances in Neural Information Processing Systems 26: 27th
  Annual Conference on Neural Information Processing Systems}, pages
  2436--2444, (2013).

\bibitem{DBLP:journals/toc/KempeKT15}
D.~Kempe, J.~M. Kleinberg, and {\'{E}}.~Tardos.
\newblock Maximizing the spread of influence through a social network.
\newblock {\em Theory of Computing}, 11:105--147, 2015.

\bibitem{DBLP:journals/ipl/KhullerMN99}
S.~Khuller, A.~Moss, and J.~Naor.
\newblock The budgeted maximum coverage problem.
\newblock {\em Information Processing Letters}, 70(1):39--45, 1999.

\bibitem{krause2005note}
A.~Krause and C.~Guestrin.
\newblock {\em A note on the budgeted maximization of submodular functions}.
\newblock Carnegie Mellon University. Center for Automated Learning and
  Discovery, 2005.

\bibitem{DBLP:journals/tec/LaumannsTZ04}
M.~Laumanns, L.~Thiele, and E.~Zitzler.
\newblock Running time analysis of multiobjective evolutionary algorithms on
  pseudo-boolean functions.
\newblock {\em {IEEE} Transactions on Evolutionary Computation}, 8(2):170--182,
  2004.

\bibitem{DBLP:conf/naacl/LinB10}
H.~Lin and J.~A. Bilmes.
\newblock Multi-document summarization via budgeted maximization of submodular
  functions.
\newblock In {\em Human Language Technologies: Conference of the North American
  Chapter of the Association of Computational Linguistics}, pages 912--920. The
  Association for Computational Linguistics, (2010).

\bibitem{nemhauser1981maximizing}
G.~L. Nemhauser and L.~A. Wolsey.
\newblock Maximizing submodular set functions: formulations and analysis of
  algorithms.
\newblock In {\em North-Holland Mathematics Studies}, volume~59, pages
  279--301. Elsevier, (1981).

\bibitem{DBLP:journals/mp/NemhauserWF78}
G.~L. Nemhauser, L.~A. Wolsey, and M.~L. Fisher.
\newblock An analysis of approximations for maximizing submodular set functions
  - {I}.
\newblock {\em Mathematical Programming}, 14(1):265--294, 1978.

\bibitem{DBLP:conf/ijcai/NeumannW15}
F.~Neumann and C.~Witt.
\newblock On the runtime of randomized local search and simple evolutionary
  algorithms for dynamic makespan scheduling.
\newblock In {\em Proceedings of the Twenty-Fourth International Joint
  Conference on Artificial Intelligence, {IJCAI}}, pages 3742--3748. {AAAI}
  Press, (2015).

\bibitem{POURHASSAN2019}
M.~Pourhassan, V.~Roostapour, and F.~Neumann.
\newblock Runtime analysis of {RLS} and {(1+1)} {EA} for the dynamic weighted
  vertex cover problem.
\newblock {\em Theoretical Computer Science}, 832:20--41, 2020.

\bibitem{qian2020multi}
C.~Qian.
\newblock Multi-objective evolutionary algorithms are still good: maximizing
  monotone approximately submodular minus modular functions.
\newblock {\em Evolutionary Computation}, pages 1--28, 2020.

\bibitem{DBLP:conf/ijcai/QianSYT17}
C.~Qian, J.~Shi, Y.~Yu, and K.~Tang.
\newblock On subset selection with general cost constraints.
\newblock In {\em Proceedings of the International Joint Conference on
  Artificial Intelligence, {IJCAI 2017}}, pages 2613--2619, (2017).

\bibitem{DBLP:journals/ai/QianYTYZ19}
C.~Qian, Y.~Yu, K.~Tang, X.~Yao, and Z.~Zhou.
\newblock Maximizing submodular or monotone approximately submodular functions
  by multi-objective evolutionary algorithms.
\newblock {\em Artificial Intelligence}, 275:279--294, 2019.

\bibitem{DBLP:conf/nips/QianYZ15}
C.~Qian, Y.~Yu, and Z.~Zhou.
\newblock Subset selection by {P}areto optimization.
\newblock In {\em Advances in Neural Information Processing Systems 28: Annual
  Conference on Neural Information Processing Systems}, pages 1774--1782,
  (2015).

\bibitem{roostapour2020pareto}
V.~Roostapour, A.~Neumann, and F.~Neumann.
\newblock Evolutionary multi-objective optimization for the dynamic knapsack
  problem.
\newblock {\em arXiv preprint arXiv:2004.12574}, 2020.

\bibitem{DBLP:conf/aaai/RoostapourN0019}
V.~Roostapour, A.~Neumann, F.~Neumann, and T.~Friedrich.
\newblock Pareto optimization for subset selection with dynamic cost
  constraints.
\newblock In {\em Proceedings of the Thirty-Third {AAAI} Conference on
  Artificial Intelligence, {AAAI 2019}}, pages 2354--2361. {AAAI} Press,
  (2019).

\bibitem{DBLP:journals/algorithmica/ShiSFKN19}
F.~Shi, M.~Schirneck, T.~Friedrich, T.~K{\"{o}}tzing, and F.~Neumann.
\newblock Reoptimization time analysis of evolutionary algorithms on linear
  functions under dynamic uniform constraints.
\newblock {\em Algorithmica}, 81(2):828--857, 2019.

\bibitem{vondrak2010submodularity}
J.~Vondr{\'a}k.
\newblock Submodularity and curvature: The optimal algorithm.
\newblock {\em {RIMS} K{\^{o}}ky{\^{u}}roku Bessatsu}, B23:253–--266, 2010.

\bibitem{DBLP:journals/ai/XuBHL07}
K.~Xu, F.~Boussemart, F.~Hemery, and C.~Lecoutre.
\newblock Random constraint satisfaction: Easy generation of hard (satisfiable)
  instances.
\newblock {\em Artificial Intelligence}, 171(8-9):514--534, 2007.

\bibitem{DBLP:conf/gecco/Yang15}
S.~Yang.
\newblock Evolutionary computation for dynamic optimization problems.
\newblock In {\em Proceedings of the Genetic and Evolutionary Computation
  Conference, {GECCO 2015}}, pages 629--649. {ACM}, (2015).

\bibitem{DBLP:conf/aaai/ZhangV16}
H.~Zhang and Y.~Vorobeychik.
\newblock Submodular optimization with routing constraints.
\newblock In {\em Proceedings of the Thirtieth {AAAI} Conference on Artificial
  Intelligence, AAAI 2016}, pages 819--826. {AAAI} Press, 2016.

\end{thebibliography}

\end{document}